\documentclass[11pt]{article}
\usepackage{geometry}                
\usepackage{graphicx}
\usepackage{amsmath}
\usepackage{amssymb}
\usepackage{epstopdf}
\usepackage{amsthm}
\usepackage{color}
\usepackage{cite}
\newtheorem{theorem}{Theorem}[section]
\newtheorem{lemma}[theorem]{Lemma}
\newtheorem{proposition}[theorem]{Proposition}
\newtheorem{corollary}[theorem]{Corollary}
\newtheorem{definition}[theorem]{Definition}
\newenvironment{remark}[1][Remark]{\begin{trivlist}
\item[\hskip \labelsep {\bfseries #1}]}{\end{trivlist}}

\newcommand{\ob}[1]{\left(#1\right)} 
\newcommand{\cb}[1]{\left[#1\right]} 

\newcommand{\E}{\mathbb{E}}

\newcommand{\R}{\mathbb{R}}

\newcommand{\set}[1]{\left\{#1\right\}} 

\newcommand{\Z}{\mathbb{Z}}

\newcommand{\eps}{\varepsilon}

\newcommand*\samethanks[1][\value{footnote}]{\footnotemark[#1]}
\newcommand{\Ec}{\mathcal{E}_{t}(t_1, \ldots, t_m)}
\newcommand{\Ecs}{\mathcal{E}_{\zeta}(t_1, \ldots, t_m)}
\newcommand{\vet}{V_{\mathcal{E}_t}}

\newcommand{\aet}{A_{\mathcal{E}_t}}

\title{Multifocality and recurrence risk:\\ a quantitative model of field cancerization}
      
\author{
Jasmine Foo${}^{1}$\thanks{Partially supported by NSF grant DMS-1224362} , Kevin Leder ${}^{2}$\samethanks, and Marc D. Ryser${}^3$\thanks{Partially supported by NIH grant R01-GM096190-02}\\
{\small 1. School of Mathematics, and 2. Industrial and Systems Engineering,}\\
{\small University of Minnesota, Minneapolis, MN} \\
{\small 3. Department of Mathematics, Duke University, Durham, NC}\\
\date{\today} }
\begin{document}
\maketitle

\begin{abstract}

Primary tumors often emerge within genetically altered fields of premalignant cells that appear histologically normal but have a high chance of progression to malignancy. Clinical observations have suggested that these premalignant fields pose high risks for emergence of secondary recurrent tumors if left behind after surgical removal of the primary tumor. In this work, we develop a spatio-temporal stochastic model of epithelial carcinogenesis, combining cellular reproduction and death dynamics with a general framework for multi-stage genetic progression to cancer. Using this model, we investigate how macroscopic features (e.g. size and geometry of premalignant fields) depend on microscopic cellular properties of the tissue (e.g.\ tissue renewal rate, mutation rate, selection advantages conferred by genetic events leading to cancer, etc). We develop methods to characterize how clinically relevant quantities such as waiting time until emergence of second field tumors and recurrence risk after tumor resection. We also study the clonal relatedness of recurrent tumors to primary tumors, and analyze how these phenomena depend upon specific characteristics of the tissue and cancer type. This study contributes to a growing literature seeking to obtain a quantitative understanding of the spatial dynamics in cancer initiation.
\end{abstract}

\section{Introduction}
The term `field cancerization' refers to the clinical observation that certain regions of epithelial tissue have an increased risk for the development of multiple synchronous or metachronous  primary tumors. This term originated in 1953 from repeated observations by Slaughter and colleagues of multiple primary oral squamous cell cancers and local recurrences within a single region of tissue \cite{slaughter1953field}.  This phenomenon, also known as the `cancer field effect' has been documented in many organ systems including head and neck (oral cavity, oropharynx, and larynx), lung, vulva, esophagus, cervix, breast, skin, colon, and bladder \cite{braakhuis2003genetic}.  Although the exact underlying mechanisms of the field effect in cancer are not fully understood, recent molecular genetic studies suggest a carcinogenesis model in which clonal expansion of genetically altered cells (possibly with growth advantages) drives the formation of a premalignant field \cite{braakhuis2003genetic,chai2009field}.  This premalignant field, which may develop in the form of one or more expanding patches, forms fertile ground for subsequent genetic transformation events, leading to intermediate cancer fields and eventually clonally diverging neoplastic growths.  The presence of such premalignant fields poses a significant risk for cancer recurrence and progression even after removal of primary tumors.  Importantly,  these fields with genetically altered cells often appear histologically normal and are difficult to detect; thus, mathematical models to predict the extent and evolution of these fields may be useful in guiding treatment and prognosis prediction.

In this work we develop and analyze a mathematical model of the cancer field effect.  This model combines spatial cellular reproduction and death dynamics in an epithelial tissue with a general framework for multi-stage genetic progression to cancer.  Using this model, we investigate how  microscopic cellular properties of the tissue (e.g.\ tissue renewal rate, mutation rate, selection advantages conferred by genetic events leading to cancer, etc) impact the process of field cancerization in a tissue. We develop methods to characterize the waiting time until emergence of second field tumors and recurrence risk after tumor resection.  We also study the clonal relatedness of recurrent tumors to primary tumors by assessing whether local field recurrences (second field tumors) are more likely than distant field recurrences (second primary tumors), and analyze how these phenomena depend upon characteristics of the tissue and cancer type. The methodology developed in this work is generally applicable to understanding the field effect in many epithelial cancers; in follow-up work we will calibrate and apply this framework to study recurrence risks due to field cancerization in specific cancer types.


The current work will contribute to a growing literature on the evolutionary dynamics of cancer initiation, see e.g.\ \cite{Armitage1957, LuMo2002, KomSenNow, Michor2006a, Schwein, IwaMicKomNow, Wodarz2007, DuScSc09, FoLeMi2011, beerenwinkel2007genetic}. In particular, there have been several previous mathematical models studying the stochastic evolutionary process of cancer initiation from spatially structured tissue, e.g.\ \cite{Ko06s, Nowak2003, WilBje72, ThLoStKo10, Ko13, DuMo10s}.  In 1977 Williams and Bjerknes proposed a model of clonal expansion in epithelial tissue \cite{WilBje72}. This model is closely related to the biased voter model from particle systems theory \cite{Liggett2}, and in \cite{BraGri81, BraGri80}  the growth properties and asymptotic shape of the process were established. Komarova studied a 1-dimensional process with mutation and showed that the probability of mutant fixation and time to obtain two-hit mutants differ from the well-mixed setting \cite{Ko06s}; this work was generalized to higher dimensions in \cite{DuMo10s}.  Later, in \cite{ThLoStKo10,Ko13} the relationships between migration, mutation, selection and invasion in a spatial stochastic evolutionary model were explored.  Martens and colleagues studied the dynamics of mutation accumulation and population adaptation on a discrete time hexagonal lattice model with nearest neighbor replacement upon reproduction in one and two dimensions \cite{martens2011interfering, martens2011spatial}.  In a recent work Antal and colleagues consider a stochastic spatial model of cancer progression where cells acquire successive fitness advantages along the edge of the tumor. In the context of this model they study the shape of the evolving tumor front as well as the number of mutations acquired in the tumor \cite{AnKrNo}.  Lastly, in previous work we studied the accumulation and spread rates of advantageous mutant clones in a spatially structured population of general dimension \cite{DuFoLe12}.  In the current work we build upon these previous studies to develop a quantitative understanding of field cancerization and tumor recurrence risks.

The article is organized as follows: in section \ref{model} we introduce the stochastic mathematical model and describe basic properties regarding the survival and growth rate of mutant clones.  Using previously derived results on the spread of mutant clones, we introduce a mesoscopic approximation to the model.  In section \ref{field_char} we analyze the model to investigate the characteristics and extent of local and distant premalignant fields at the time of initiation.  In particular, we determine how the spatial geometry of the field (e.g.\ number and size of lesions) depends on cellular and tissue properties such as mutation rate, tissue renewal rate and mutational fitness advantages. In section \ref{recur_pred} we analyze the model to understand the risk of recurrence due to local or distant field malignancies, as a function of time and cellular parameters.

Throughout the paper we will use the following notation for the asymptotic behavior of positive functions, 
\begin{align*}
f(t) \sim g(t) & \quad\hbox{if $f(t)/g(t) \to 1$ as $t \to \infty$}, \\
f(t) = o(g(t)) \ \hbox{or} \ f(t)\ll g(t) & \quad\hbox{if $f(t)/g(t) \to 0$ as $t \to \infty$}, \\
f(t) \gg g(t) & \quad\hbox{if $f(t)/g(t) \to \infty$ as $t \to \infty$}, \\
f(t) = O(g(t)) & \quad \hbox{if $f(t) \leq C g(t)$ for all $t$}, \\
f(t) = \Theta(g(t)) & \quad \hbox{if $c g(t) \le |f(t)| \leq C g(t)$ for all $t$} .
\end{align*}
Finally, we use the notation $X=_d F$ to denote that the random variable $X$ has distribution $F$.

\section{Mathematical framework and basic properties}\label{model}
Cancer initiation is associated with the accumulation of multiple successive genetic or epigenetic alterations to a cell \cite{weinberg2013biology}. A subset of these genetic events may give rise to a fitness advantage (i.e.\ an increase in reproductive rate of the cell or avoidance of apoptotic signals), and subsequently lead to a clonal expansion within the tissue.  The expanding mutant cell populations form the background for further independent genetic events which lead to carcinogenesis.  
As a result of this spatial evolutionary process, by the time of cancer initiation or diagnosis the tissue field surrounding a tumor can be composed of genetically distinct premalignant lesions of various sizes and stages.


\subsection{Cell-based model}

To study the dynamics of this process, we consider a stochastic model which describes the accumulation and spread of a clone of cells with genetic alterations throughout a spatially structured tissue (e.g.\ stratified epithelium). Thus, we consider the model on a regular lattice $\mathbb{Z}^d\cap [-L/2,L/2]^d$, where $L>0$ and $d$ is the number of spatial dimensions of the tissue. Each location in the lattice is occupied by a single cell, and each cell reproduces at a rate according to its fitness with exponential waiting times.  Whenever a cell reproduces, its offspring replaces one of its $2d$ lattice neighbors at random, see Figure \ref{fig:growth}A.  The type of each cell corresponds to its fitness, which is related to the number of genetic hits a cell has accumulated in a multi-step genetic model of cancer initiation.  For example, type-0 cells have fitness normalized to 1 and are labeled as wild-type or normal (with no mutations).  Initially our entire lattice is occupied by type 0 cells. Type-0 cells acquire the first mutation at rate $u_1$ to become type-1 cells. The type-1 cell will have a relative fitness advantage to type-0 cells, given by $1+s_1$, for some constant $s_1 \geq 0$.  In general, type-$i$ cells have a fitness advantage of $1+s_{i}$ relative to type-$(i-1)$ cells, and they acquire the $(i+1)-$th mutation in the sequence at rate $u_{i+1}$ to become type-$(i+1)$ cells. The process is stopped when a cell develops $k$ mutations; we call this the time of cancer initiation. The number of mutation $k$ as well as the parameters $u_i, s_i$ for $i=1, \ldots, k$  depend on the specific cancer type. Although many (epi)genetic events are selectively disadvantageous (i.e.\ they confer a selective disadvantage $s_i<0$), the progeny of deleterious mutants die out quickly so here we restrict our attention to the case $s_i\geq0$.  Note that this process can be thought of as a spatial version of the Moran process, a spatially well-mixed population model that is commonly used to describe carcinogenesis (e.g.\ see \cite{Schwein, IwaMicKomNow, Wodarz2007, DuScSc09, FoLeMi2011}).  In addition, the spatial reproduction and death dynamics of this model (without mutation) correspond to the biased voter process which has been well-studied in physics and probability literature.  In fact, a similar voter model approach was previously used to model cellular dynamics in epithelial tissue and found to correlate well with experimental predictions of clone size distribution in the mouse epithelium \cite{Klein}.

 The total number of cells in the fixed-size population is $N\equiv L^d$; in most cancer initiation settings this number is quite large (at least $10^6$), while mutation rates are quite small (orders of magnitude smaller than 1). Therefore we will, unless stated otherwise, restrict our analysis to regimes where $L\gg1$ and $u_i\ll1$.  In Section \ref{mes_model}, we will briefly discuss  the specific conditions that we impose on the relationship between these parameters.  For mathematical simplicity, the lattice is equipped with periodic boundary conditions; however in most relevant biological situations the domain size (i.e.\ cell number) is sufficiently large so that boundary effects are negligible.

\begin{figure}[htbp] 
   \centering
   \includegraphics[width=.7\textwidth]{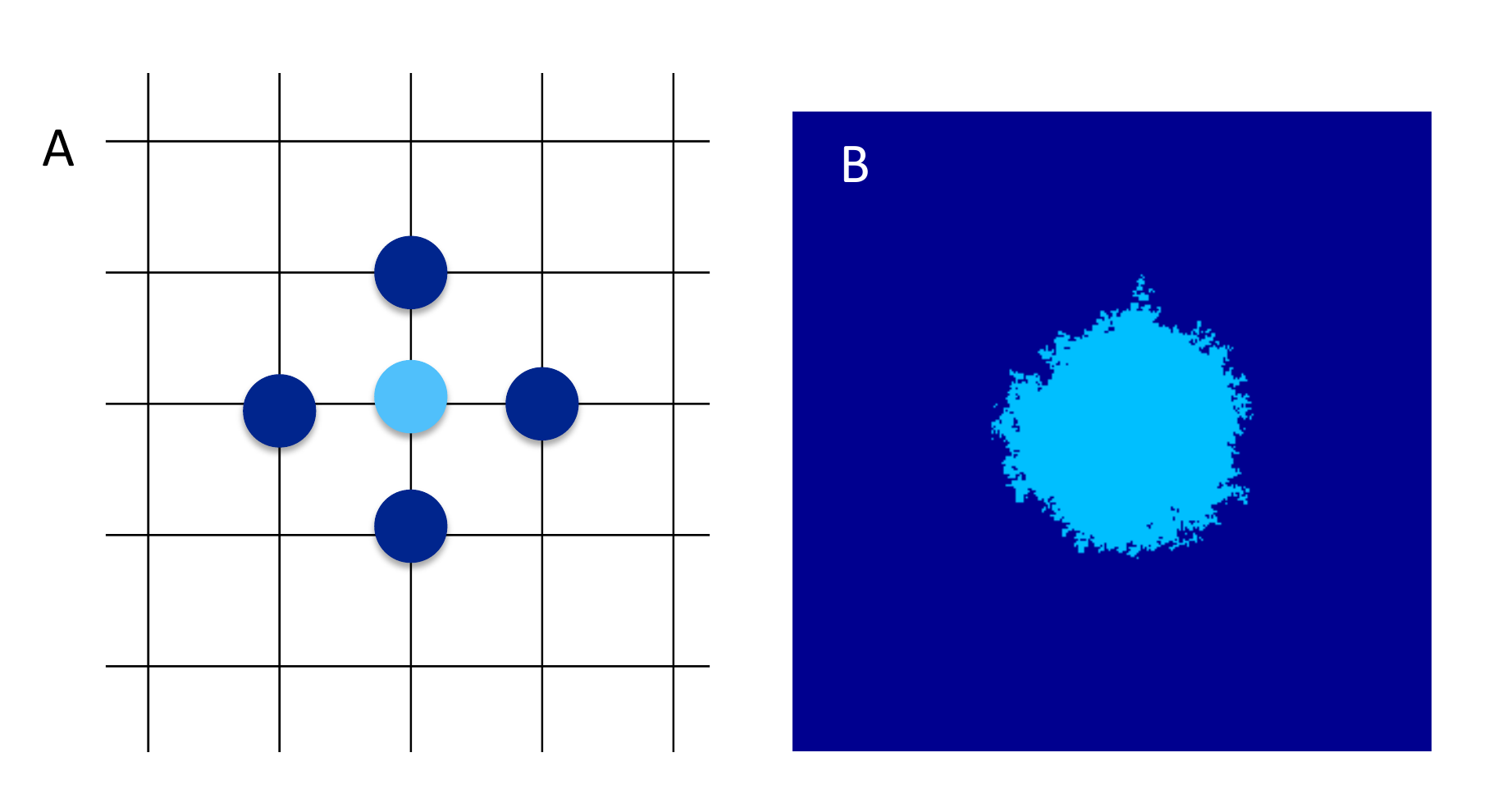} 
   \caption{{\bf Lattice dynamics.} (A) Schematic of spatial Moran model in $d=2$: each cell divides at rate according to its fitness and replaces one of its 2$d$ neighbors: if the light blue cell divides, its offspring replaces one of the dark blue neighbors, chosen uniformly at random. Every lattice site is occupied at all times (not shown). (B) Simulation example of the model: growth of an advantageous clone (light blue) starting from one cell with fitness advantage $s=0.2$ over the surrounding field (dark blue). }
     \label{fig:growth}
\end{figure}

{\it Note on dimension of the model.}  
We analyze the general model in space dimensions $d=1,2,3$. While all epithelial tissues have an intrinsically three dimensional architecture, in some situations considering $d=1,2$ may be a good approximation. For example, cancer initiation in mammary ducts of the breast, renal tubules of the kidney, and bronchi tubes of the lung could be viewed as approximately one-dimensional processes, due to the aspect ratio of tube radius versus length. On the other hand, cancer initiation in the squamous epithelium of the cervix, the bladder or the oral cavity can be viewed as two-dimensional process, since initiation occurs in the basal layer of the epithelium which is only 1-2 cells thick (e.g.\ see Figure \ref{fig:epi_cartoon}). The validity of such approximations poses an interesting problem in itself, but will not be addressed in this work. 
 \begin{figure}[htbp] 
   \centering
   \includegraphics[width=.8\textwidth]{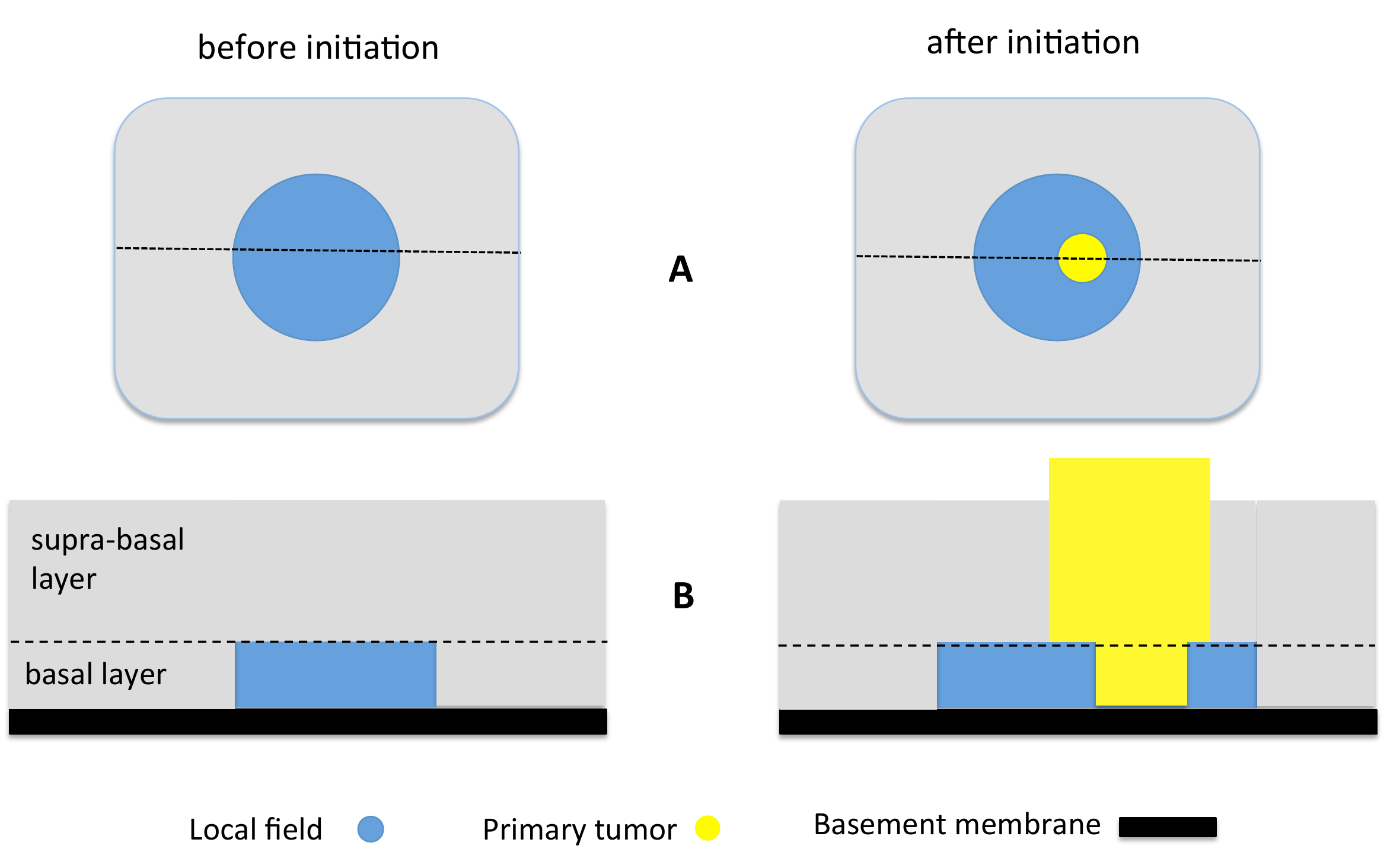} 
   \caption{{\bf Geometry of squamous epithelium.} {\bf A} Basal layer (vertical perspective) before initiation with local field (left), and after initiation where the tumor is growing within the local field (right). {\bf B}  Sideways view of the fields before and after initiation, along the dashed lines in panel A. The proliferative cells inhabiting the two-dimensional lattice in the model reside in the basal layer of the epithelium. After initiation, malignant growth is not restricted to basal layer only.}
   \label{fig:epi_cartoon}
\end{figure}


\subsection{Survival and growth of a single mutant clone}\label{surv}
We first establish some basic behaviors of mutant cells and their clonal progeny within a tissue.  Of particular interest are: (i) the survival probability of a mutant clone, and (ii) the rate of spatial expansion of the mutant clone through the tissue.  In particular, how are these characteristics influenced by tissue parameters and the cellular fitness advantage conferred by a mutation? We have addressed some of these questions in a previous work \cite{DuFoLe12} and restate the results here to make the paper self-contained.  In addition, we perform new simulations in this work to fill in gaps where theoretical results are currently not available.   

Consider the probability that a mutant cell survives to form a viable clone (i.e.\ does not die out due to demographic stochasticity).  Let type-1 cells have fitness $1+s$ and type-0 cells have fitness $1$, and let $\phi_t(x)$ denote the type of cell at site $x$ in the lattice at time $t$.  Define $$\xi_t \equiv \{ x \in \mathbb{Z}^d\cap [-L/2,L/2]^d: \phi_t(x) = 1 \}.$$  In other words, $\xi_t$ is the set of all type-1 cell locations at time $t$.  
We initiate the model with a single type-1 cell at the origin surrounded by type-0 cells in all other locations:
$$
\phi_0(x)=\begin{cases} 1,&\enskip x=0\\ 0,&\enskip\mbox{otherwise,}\end{cases}
$$ 
and assume no further mutations are possible ($u_i = 0$).
This simplified model is known as the Williams-Bjerknes model \cite{WilBje72}, and if   $L=\infty$ then it corresponds to the biased voter model, see e.g.\  \cite{liggett2005interacting}.
Let $|\xi_t|$ denote the number of type-1 cells in the model at time $t$.  Then we  can define the extinction time of the process $T_0\equiv\inf\{t>0: |\xi_t|=0\}$.
The probability of survival of a single mutant clone with selective advantage $s$ over the surrounding cells is then the probability of the event $\{T_0=\infty\}$. By looking at the the process $|\xi_t|$ only at its jump times, we note that the embedded process is a discrete time random walk that moves one up with probability $s/(1+s)$ or one down with probability $1/(1+s)$.  This can be seen by observing that the process only changes at boundaries between type-0 and type-1 cells, and the only possible resulting events are that the type-0 gets replaced by a type-1 (resulting in a jump up in $|\xi_t|$) or the type-1 gets replaced by a type-1 (resulting in a jump down in $|\xi_t|)$.  Analysis of the overall survival probability of this random walk can then be calculated using simple gambler's ruin analysis, 
$$
P(T_0=\infty)=\frac{s}{1+s}\approx s,
$$
where the approximation is valid for $s\ll1$.  Thus, the probability that a mutant clone with fitness advantage $s$ survives is $\frac{s}{1+s}$, and is independent of the dimension of the tissue.

It is important to understand how the expansion rate of a mutant clone depends on the selection strength $s$ of the mutant.  To this end, we first state a  result proven by Bramson and Griffeath \cite{BraGri80,BraGri81} which establishes an asymptotic shape for the type-1 clone.
\begin{theorem}[Bramson and Griffeath]\label{BGthm}
For any $\eps>0$, there is a $t_\eps(\omega)$ such that on $\{ T_0 = \infty\}$ 
$$
(1-\eps)tD \cap \Z^d \subset \xi_t \subset (1+\eps) tD \quad \hbox{for $t \ge t_\eps(\omega)$,}
$$
where $D$ is a convex set and has the same symmetries as $\Z^d$. 
\end{theorem}
In other words, the Bramson-Griffeath shape theorem says that conditional on the clone never going extinct,  the radius of the clone expands linearly.  In previous work, we studied how this linear rate of expansion depends on the selection strength $s$ in the setting of weak selection strength, see Theorem 1 of \cite{DuFoLe12}. More precisely, if we denote by $e_1$  the first unit vector in $\R^d$ and define the growth rate $c_d(s)$ such that $$D\cap\{ze_1:z\in\mathbb{R}\}=[-c_d(s),c_d(s)],$$ then as $s\to 0$,
\begin{align}\label{s_asymp}
c_d(s) \sim \begin{cases} s & d=1 \\ \sqrt{4\pi s/\log(1/s)} & d=2 \\ \sqrt{4\beta_3s} & d = 3, \end{cases}
\end{align}
where $\beta_3$ is the probability that two simple random walks started at 0 and $e_1=(1,0, 0)$ never hit.
In other words, the radius of the asymptotic shape $D$ approximating the type-1 clone grows linearly with rate on the order of $c_d(s)$.   Note that as the dimension $d$ increases, the growth rate increases due to a larger clone boundary size in higher dimensions.

The previous results hold only in the regime of weak selection or small $s$. For larger values of the selective advantage $s$, simulations can be used to obtain $c_d(s)$ for $d=2,3$  (in $d=1$ the process can be analyzed directly through simple random walk analysis and we obtain that $c_1(s)=s$). For example, Figure \ref{fig:small_s} shows that the $s$-dependence of the growth rate is approximately linear for $s> 0.5$; in this case simple regression yields the estimate $c_2(s)\approx 0.6 s + 0.22$ ($s>0.5$).  Thus, a combination of analysis and simulation gives us a complete picture of how spatial expansion rate of mutant clones in a tissue depend upon the selective advantage $s$ for a wide range of selection strengths.
\begin{figure}[htb!]
   \centering
   \includegraphics[width=0.5\textwidth]{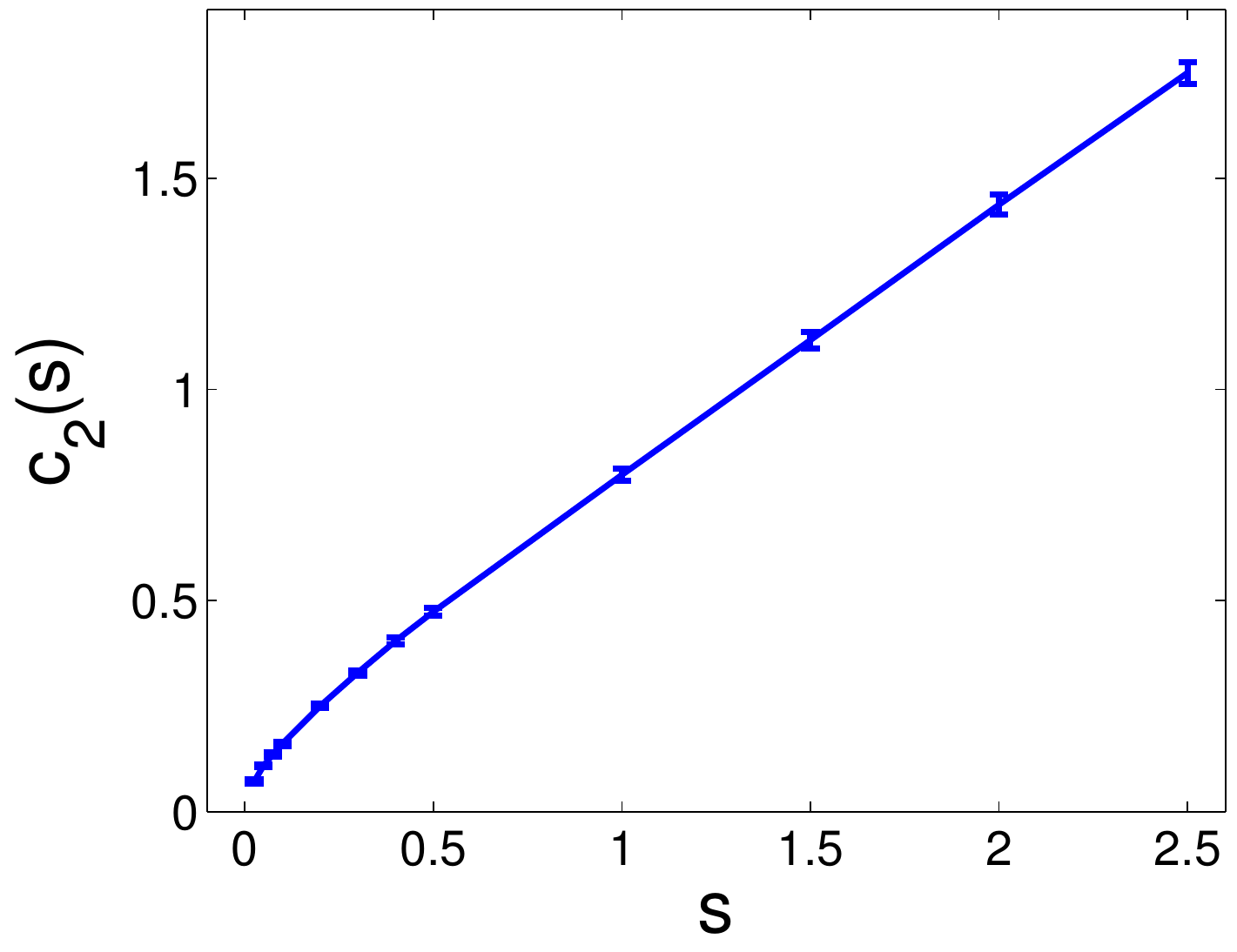} 
   \caption{{\bf Simulations of clonal expansion rate for large $s$.} Dependence of the growth rate $c_2$ on the fitness advantage $s$. Statistics performed on  $M=100$ samples for each $s$-value. 
   }  
   \label{fig:small_s}
\end{figure}

\subsection{Approximating with a hybrid mesoscopic model}\label{mes_model}

Our results regarding the survival and growth of a single mutant clone suggest a hybrid mesoscopic model simplification that enables our analysis of the field cancerization process.  In particular, each successful mutant clone can be well-approximated as growing $d$-dimensional ball with expansion rate $c_d(s)$ as calculated in the previous section.  Before proceeding however, let us clarify the notion of clone `survival' a.k.a. `success' in the full model, where multiple mutations can arise and compete in the same finite domain.  In particular, we consider a mutant clone with selective advantage $s$ over the background to be successful if it reaches size $\gg 1/s$. This criterion guarantees a negligible chance of extinction in an infinite domain with no interference.  In particular, if we define $T_0$ to be the extinction time of a biased voter model with selective advantage $s$, one can use the embedded discrete time process and gambler's ruin calculations to show that if the biased voter process reaches $k\gg 1/s$ then $P\left(T_0=\infty\big | |\xi_0|=k\right)\approx1-e^{-ks}.$

Consider the fate of an unsuccessful type-1 clone arising on a background of type-0 cells. The clone evolves as a supercritical ($s > 1$) biased voter model conditioned on extinction. In \cite{DuFoLe12} we showed that unsuccessful type-1 mutations typically die out by a time of order
\begin{align}
\ell(s)= \begin{cases} s^{-2} & d=1,\\
s^{-1} \log \ob{1/s} & d=2, \\
s^{-1} & d=3.
 \end{cases}
\end{align}
As seen in the previous section, the survival probability in the biased voter model (starting with a single type 1 cell in a sea of type 0 cells) is $s/(1+s)$, but in the more complex spatial Moran model with the possibility of multiple interacting type 1 clones, it is not immediately clear that this survival probability is still given by $s/(1+s)$. However, it was shown in \cite{DuFoLe12} that the above survival probability remains a good approximation as long as 
\begin{align}
(A0) \qquad (1/u_1) \gg \ell(s)^{(d+2)/2}.
\label{A0}
\end{align}
 If the total number of type-1 cells is always a negligible fraction of $N$ and (A0) holds, then successful type-1 mutations arrive as a Poisson arrival process with approximate rate $Nu_1 \frac{s}{s+1}$, where $N$ is the total number of cells in the tissue. In particular, these conditions hold for all the numerical examples in this article. 

We are now ready to introduce a hybrid mesoscopic model approximation as follows: Type-1 mutations arrive in the healthy tissue as a Poisson arrival process with rate $Nu_1$, distributed uniformly at random in the spatial domain.  Each mutation event has two potential outcomes:
\begin{itemize}
\item with probability $s/(1+s)$, the mutation is successful and we approximate the subsequent clonal expansion with a ball whose radius grows deterministically.  The macroscopic growth rate is $c_d(s)$, which was derived from individual cellular growth kinetics as described in section \ref{surv}.  As a representative simulation in figure \ref{fig:growth}B suggests, the ball in standard $L^2$-norm in $\R^d$ will be utilized.  
\item with probability $1/(1+s)$, the mutation is unsuccessful, and the clone evolves according to the full stochastic (cellular-level) model dynamics conditioned on extinction.  
\end{itemize}
Note that the remainder of the paper discusses properties of this mesoscopic model.

It will be useful to define $\gamma_d$ as the volume of a ball of radius 1 in $d$ dimensions, 
$$
\gamma_1=2, \qquad \gamma_2=\pi, \qquad \gamma_3 = 4\pi/3.
$$
Note that although the stochastic fluctuations of the shape of expanding clones are lost in this approximation, one gains generality since the mesoscopic model can approximate a whole class of microscopic models that admit a shape result in the line of Theorem \ref{BGthm}.   

\subsection{Cancer initiation behavior}

Although the methodology developed in this work can be generalized to the setting of $k$-mutation carcinogenesis models, we will consider for simplicity the classic two-mutation model of cancer initiation first introduced by Knudson \cite{Knudson2001}. Here, type-0 cells are wild-type with fitness $1$, type-1 cells are premalignant with fitness $1+s_1$ relative to type-0 cells, and type-2 cells are initiated cancer cells with fitness $1+s_2$ relative to type-1 cells.  The time of cancer initiation $\sigma_2$ is defined as the time at which the first successful type-2 cell arrives.  In \cite{DuFoLe12}, we studied the situation where $s_1 = s_2 = s > 0$ and found that the timing of cancer initiation is strongly governed by the limiting value of the following meta-parameter:
$$
\Gamma \equiv (Nu_1 s)^{d+1} (c_d^d u_2s)^{-1}.
$$
Roughly speaking, $\Gamma^{1/(d+1)}$ represents the ratio of the rate of producing successful type-1 cells to the subsequent time it take to acquire the first successful type-2.
We found that both the mechanisms and distribution of the cancer initiation time vary significantly depending on the regime of $\Gamma$:
\begin{itemize}
\item Regime 1 (R1): When $\Gamma <1$, the first successful type-2 mutation occurs within the expanding clone of the first successful type-1 mutation (left panel of Figure \ref{fig:regimes}). The initiation time $\sigma_2$ is exponential and does not depend on the spatial dimension.
\item Regime 2: (R2) For $\Gamma \in (10, 100)$, the first successful type-2 mutation occurs within one of several successful type-1 clones (middle panel in Figure \ref{fig:regimes}). The initiation time is no longer exponential and depends explicitly upon the spatial dimension.
\item Regime 3 (R3): When $\Gamma >1000$, the first successful type-2 mutation occurs after many successful type-1 mutations have occurred (right panel of Figure \ref{fig:regimes}). The first successful type-2 can arise from {\it either} a successful {\it or} an unsuccessful type-1 family; the initiation time represents a mixture distribution of these two events.  
\item Note that for $\Gamma\in [1,10]$ and $\Gamma\in [100,1000]$ we say that we are in borderline regimes R1/R2 and R2/R3 respectively.
\end{itemize}

We refer the reader to  \cite{DuFoLe12} for mathematical details of these statements.  Note that these `regimes' can be thought of as labels highlighting distinct types of initiation behaviors that arise as $\Gamma$ changes. In fact the system behavior continuously varies through the parameter space, and borderline cases between these regimes do exist. Figure \ref{fig:wait_time} shows how the distribution of the waiting time $\sigma_2$ varies with changing number of cells $N$ in $d=2$.  We note that as $N$ increases, the waiting time distribution shifts to the left and initiation occurs earlier.   By comparing Figures \ref{fig:regimes} and  \ref{fig:wait_time} we see that early initiation times are associated with a diffuse premalignant field with a large  number of independent lesions, whereas late initiation times are associated with a single premalignant field harboring the initiating tumor cell.
\begin{figure}[htbp] 
   \centering
  \includegraphics[width=1\textwidth]{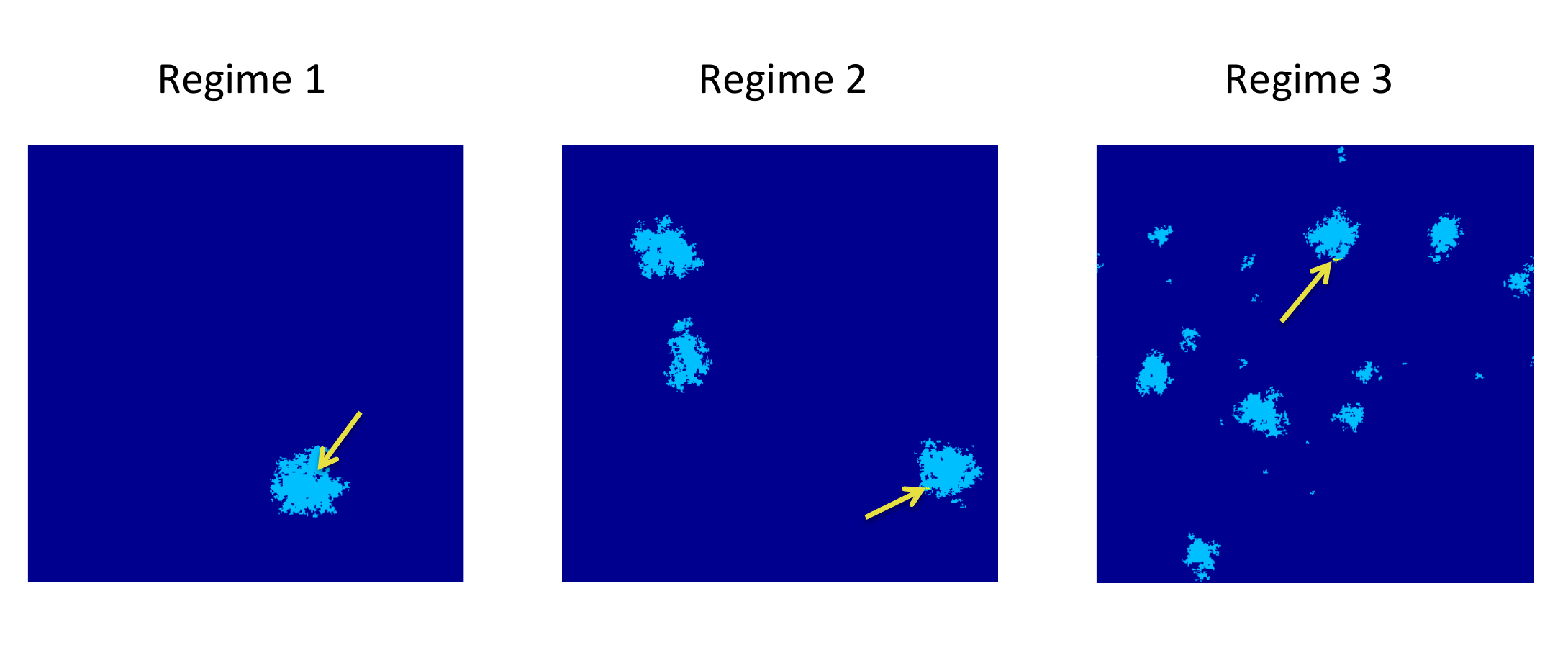} 
   \caption{ {\bf The three dynamic regimes.} {\it Regime 1:} first successful type-2 cell (arrow) arises in the first premalignant clone, $\Gamma=0.055$. {\it Regime 2:} several premalignant clones are present at the time of the first successful type-2 cell, $\Gamma = 54.47$. {\it Regime 3:} a large number of small premalignant clones  are present by the time of the first successful type-2 cell, $\Gamma=5.45\times 10^4$. Simulations obtained with parameter values as in Figure \ref{fig:wait_time}.}
   \label{fig:regimes}
\end{figure}

\begin{figure}[htbp] 
   \centering
   \includegraphics[width=.5\textwidth]{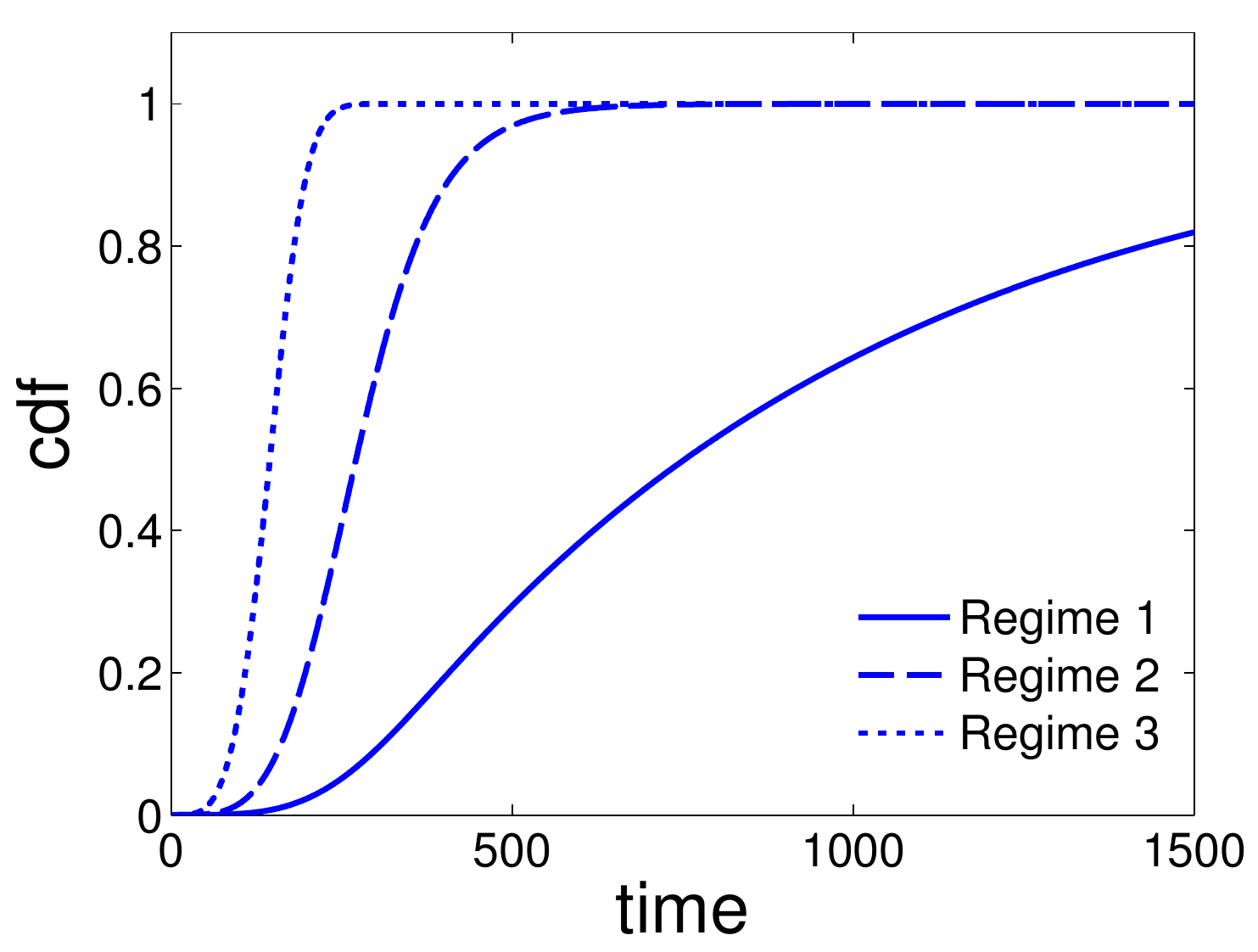} 
   \caption{{\bf Waiting  time until first successful type-2.} Cumulative distribution function (cdf) of $\sigma_2$, the waiting time until the first successful type-2 mutation, for increasing $N$ (see Theorem \ref{pdfsig2}). Regime 1: $u_1=7.5\cdot 10^{-8}$, Regime 2: $u_1=7.5\cdot 10^{-7}$, Regime 3: $u_1=7.5\cdot 10^{-6}$.  All other parameters are fixed: $d=2$, $N=2\cdot 10^5$, $s_1=s_2=0.1$, $u_2 = 2\cdot 10^{-5}$, $c_2(s_1) = 0.16$.}
   \label{fig:wait_time}
\end{figure}

To briefly summarize this section, we have described first a microscopic model of cellular division, mutation and death within a regularly structured epithelial tissue.  Analysis of the fine-scale dynamics of this model leads to a more tractable hybrid mesoscopic model which approximates the microscopic model.  In a previous work, we studied the time until cancer initiation ($\sigma_2$) under this hybrid model and we found three distinct regimes of initiation behavior in which the distribution of $\sigma_2$ has distinct parametric forms.  In the next section, we analyze this mesoscopic model to study the characteristics and extent of premalignant fields at the stochastic time of cancer initiation or diagnosis.  In the analyses throughout, we will consider parameter ranges spanning all three regimes of initiation behavior; however, for simplicity in regime 3 we will restrict ourselves to the range of parameter space in which successful type-2 mutations arise from successful type-1 mutations (i.e.\ that do not later die out).  The behavior in the final remaining portion of the parameter space in regime 3 will be the subject of further work.

\section{Characterizing the premalignant field}
\label{field_char}

The time between cancer initiation and diagnosis, referred to here as $T_D$, is a subject of great interest, see e.g.\ \cite{attolini2009evolutionary} for a review. In general, $T_D$ is itself a random variable and may depend on the natural history of the disease until initiation. However, if we assume that $T_D$ is independent of $\sigma_2$, then we can characterize the premalignant field at time of diagnosis, $\sigma_2+T_D$, by means of the field characterization at time $\sigma_2$, together with the distribution of the delay time $T_D$. For this reason, even though the clinically relevant time is $\sigma_2 + T_D$, we focus here on characterizing the field at $\sigma_2$. Note that mathematically, this requires us to condition our analyses upon observing $\sigma_2$ at some time $t$, i.e.\ condition upon the event $\{\sigma_2 =t \}$.

The starting time of the model ($t=0$) is assumed to be at the end of tissue development and the start of the tissue renewal phase. However for some tissues it is difficult to estimate this time, and thus 
it may be difficult to ascertain the system time $t$ at the time $\sigma_2$.  In such cases, it is simple to adapt our analyses to this scenario and treat $\sigma_2$ as an unknown quantity, by  removing the conditioning on $\{\sigma_2 =t\}$ and integrating of our results against the density of $\sigma_2$. The following theorem provides this density in terms of
\begin{align}\label{lambda}
\lambda\equiv N u_1 \frac{s}{1+s},
\end{align}
which is the arrival rate of successful type-1 mutations, and
\begin{align}\label{phi}
\phi(t) \equiv \frac{1}{t}\int_0^t \exp\ob{-\theta r^{d+1}} dr,
\end{align}
where
\begin{align}\label{theta}
\theta \equiv \frac{u_2 \bar{s}_2 \gamma_d c_d^d(s_1)}{d+1},
\end{align}
and $\bar s_2=s_2/(1+s_2)$.

\begin{theorem}\label{pdfsig2}
 The waiting time $\sigma_2$ until birth of the first successful type-2 mutant, has   probability density function
\begin{align*}
\lambda e^{t\lambda(\phi(t)-1)}  \ob{1- e^{{-\theta  t^{d+1}}}}.
\end{align*}
\end{theorem}
\begin{proof}
See (\ref{dens_sig2}) in section \ref{app_loc_field}.
\end{proof}

\subsection{Size of the local field at initation}\label{loc_field_sec}
We are first interested in characterizing the size of the local field, i.e.\ the region of the premalignant type-1 clone that gives rise to the first successful type-2 clone (see Figure \ref{SFPT}).  Following the nomenclature of \cite{braakhuis2002second}, we note the distinction between two different types of recurrent tumors: if the recurrence arises from a transformed cell in the premalignant field that gave rise to the primary tumor, the recurrence is called a {\it second field tumor}, see Figure \ref{SFPT}A. On the other hand, if the recurrence arises from a premalignant field that is clonally unrelated to the primary malignancy, it is called a {\it second primary tumor}, see Figure \ref{SFPT}B.  
These two types of recurrent tumors vary in terms of their degree of clonal relatedness to the primary tumor, and this may have some implications for treatment strategies in primary vs. recurrent tumors.

\begin{figure}[htbp]
   \centering
   \includegraphics[width=1\textwidth]{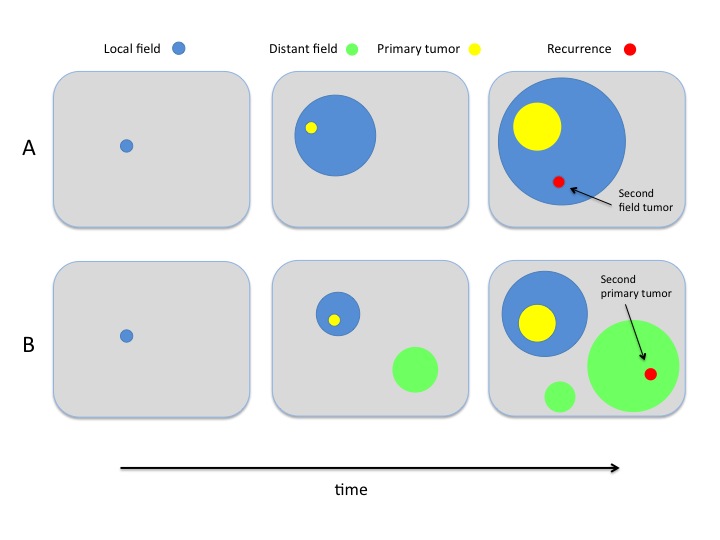} 
   \caption{{\bf Local and distant recurrences.} Local (blue) and distant (green) premalignant fields give rise to second field tumors and  second primary tumors (both red), respectively. In scenario {\bf A}, there is only one premalignant field (the local field) present at time of cancer initiation (middle panel), and the recurrence occurs inside the local field. In scenario {\bf B}, two unrelated precancerous fields are present at time of initiation (middle panel), and the recurrence may occur as a second primary tumor in the distant field.}
   \label{SFPT}
\end{figure}

We define now $R_l(t)$ to be the radius of the local field at time $t$, and $X_l(t)$ its corresponding area ($X_l=\gamma_d R_l^d$). Note that we will use the terminology `area' do describe clone sizes in all dimensions, and reserve the use of the term `volume' for space-time quantities. In the following, we are interested in determining the distributions of these two quantities at time $\sigma_2$, conditioned on the event $\{\sigma_2=t\}$. In other words, we are looking for the distributions of $\ob{R_l(\sigma_2)| \sigma_2 =t}$ and $\ob{X_l(\sigma_2)| \sigma_2 =t}$, respectively.  

At any given time, each clone produces initiating mutations at a rate proportional to its area. Hence the probability that clone $i$ (born at time $T_i$) gives rise to the initiating mutation at time $t$ is given by the ratio of clone $i$'s own area, $$X_i(t) \equiv \gamma_dc_d^d(s_1)(t-T_i)^d,$$ divided by the total area of type-1 clones present. In other words, the size distribution of the initiating clone is given by the distribution of a {\it size-biased pick} from the different clones present at the time the initiated mutation arises.  
\begin{definition}[Size-biased pick]\label{sbp}
Let  $L_1, \ldots, L_n$  be a family of $n$ random variables. A size-biased pick from $L_1, \ldots, L_n$ is defined as a random variable $L_{[1]}$ with conditional probability distribution
 $$
P(L_{[1]} =L_i |L_1,\ldots,L_n)= L_i/\sum_{j=1}^nL_j.
$$
\end{definition}
The following theorem is the main result of this section and characterizes the size-distribution of the local field at the time of initiation.  This is recognized as a size-biased pick from the clones present at time $t$, conditioned on the event $\{\sigma_2=t\}$.
  \begin{theorem}\label{sbp_thm}
 The distribution of the area of the local field at time $\sigma_2$, conditioned on $\{\sigma_2  \in  dt\}$, is  given by
 \begin{align}\label{dist1}
\hat{P}\ob{X_l(\sigma_2) \in dx}=  \hat{P}\ob{X_{[1]}\in dx} =  \frac{u_2 \bar{s}_2x^{1/d} }{d \gamma_d^{1/d} c_d(s_1)  (1-e^{-\theta t^{d+1}})} \exp\left[\frac{-u_2\bar{s}_2 x^{\frac{d+1}{d}}}{(d+1)\gamma_d^{1/d} c_d(s_1) } \right],
 \end{align}
 for $x \in [0, \gamma_d c_d^d(s_1) t^d]$. 
   \end{theorem}
   
     \begin{proof}
  See section \ref{app_loc_field}. 
  \end{proof}
  
   The distribution of the local field radius follows now easily. 
 \begin{corollary}\label{field_radius}
  The radius of the local premalignant field at $\sigma_2$, conditioned on the event $\{ \sigma_2 =t\}$, has density
$$
 \hat{P}(R_l(\sigma_2) \in dr) =   \frac{  u_2 \bar{s}_2\,\gamma_d\, r^d}{c_d(s_1)(1-e^{-\theta t^{d+1}})} \exp\left[-\frac{u_2\bar{s}_2\gamma_dr^{d+1}}{c_d(s_1)(d+1)}\right],
$$
for $r\in [0,c_d(s_1) t]$.
 \label{Cor_r}
 \end{corollary}

 \begin{figure}[htb!]
   \centering
   \includegraphics[width=1\textwidth]{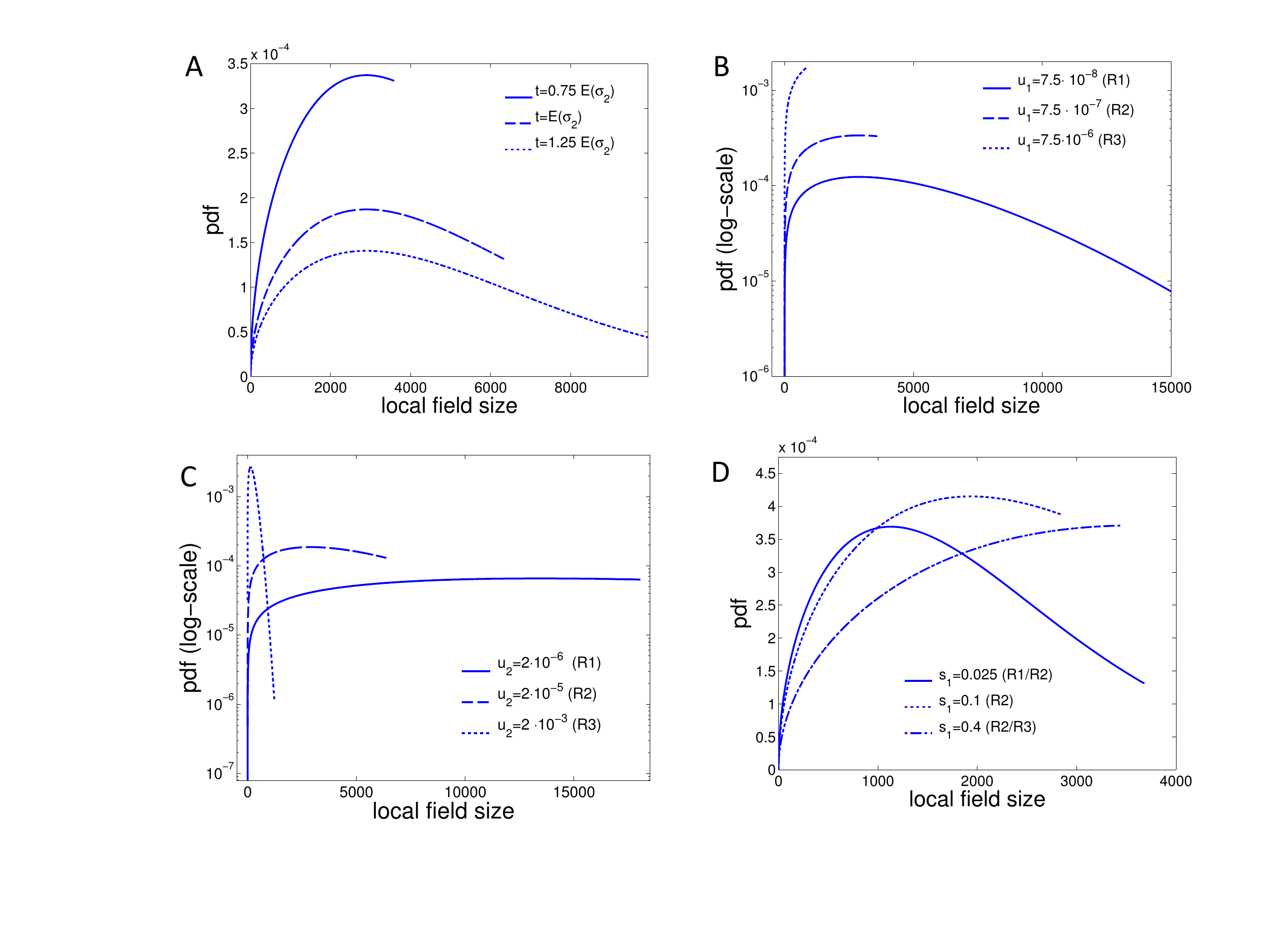} 
   \caption{ {\bf Size-distribution of local field.}  The size-distribution (\ref{dist1}) of the local field is shown for different scenarios, corresponding to the three regimes R1, R2 and R3 illustrated in Figure \ref{fig:regimes}. {\bf A} For varying arrival times $t$; {\bf B} for varying type-1 mutation rates $u_1$; {\bf C} for varying type-2 mutation rates $u_2$; ({\bf D}) for varying type-1 fitness advantages $s_1$. The non-varying parameters are held constant at  $d=2$, $N=2\cdot 10^5$, $u_1=7.5\cdot 10^{-7}$, $u_2=2\cdot 10^{-5}$, $s_1=s_2=0.1$ and $c_2(s_1)=0.16$.  } 
   \label{fig:sbp_compilation}
\end{figure}

Note that the distribution of the local field size (\ref{dist1}) depends on the rate of successful mutations  $u_2\bar s_2$ and the growth rate $c_d(s_1)$, but is independent of $\lambda$, the arrival rate of type-1 mutations. In Figure \ref{fig:sbp_compilation}A, we show how the distribution of the local field area (\ref{dist1}) changes with arrival time of the first successful type-2 clone. As expected, the support of the distribution increases with increasing initiation time, and hence the likelihood of having a large local field increases substantially. This suggests that that tumors appearing later have a higher recurrence probability if only the malignant portion is removed during surgery. The finite support of each probability density function reflects the fact that there is a hard upper bound on the size of a premalignant field at finite time $t$ in the system.

In Figure \ref{fig:sbp_compilation}B,C we illustrate the sensitivity of the size-distribution of the local field to varying mutation rates $u_1$ and $u_2$, conditioned on observing initiation at the expected time $t = E(\sigma_2)$. The mutation rates are tuned to vary across parameter Regimes 1, 2, and 3 as described in the previous section.  Observe that for lower mutation rates, the local field size varies widely (and sometimes close to uniformly) over a large range of values, while elevated mutation rates in both cases signify smaller local fields. For the $u_1$ rate (Figure 7B), an intuitive explanation for this behavior is that as the mutation rate increases, the system moves towards regimes 2 and 3, in which the premalignant field is comprised of an increasing number of independent type-1 patches. The necessary number of man-hours to the first successful type-2 mutation is absolved faster, and hence the size of the  patch that gives rise to the first successful type-2 decreases accordingly. For $u_2$ (Figure 7C) on the other hand, an increase in the mutation rate signifies a move towards regime 1: fewer type-1 clones are required to produce the first successful type-2, and the size of the type-1 field that yields the first type-2 decreases with increasing $u_2$.  Another observation to note is that the local field size varies across the same range of orders of magnitude as the mutation rates. This suggests for example, that carcinogen exposure or environmental causes changing mutation rates by one order of magnitude could result in predicted field sizes impacted similarly by an order of magnitude.

 Finally, we demonstrate the sensitivity of the local field size to the selective advantage $s$ of mutant cells, see Figure \ref{fig:sbp_compilation}D. For a small fitness gain of $s=0.025$, the distribution is peaked at lower field sizes, but as $s$ increases the field size distribution shifts to the right.  High fitness gains are usually  associated with an aggressive tumor phenotype, and Figure \ref{fig:sbp_compilation}D suggests that such tumors may also be associated with large surrounding premalignant fields and thus higher recurrence risks.

\subsection{Size of the distant field at initiation}\label{distfieldinit}
Next we are interested in analyzing the size distribution of the distant field at initiation, which is comprised of premalignant clones that are clonally unrelated to the tumor. Define the vector of areas of the distant premalignant lesions at time $t$ to be $\bar{X}_d(t)$. This vector holds the areas of all premalignant clones except for the local field clone from which the tumor arises. Mathematically speaking, the goal of this section is to characterize the law of $\bar{X}_d(\sigma_2)$ conditioned on the event $\{\sigma_2 =t\}$. Before stating the main result some additional notation is needed. First, define the mapping $\alpha_j(i)$  as follows:
$$
\alpha_j(i) =
\begin{cases}
i, & \mbox{if }  j >i\\
i+1, &\mbox{if } j \leq i. 
\end{cases}
$$
Then, we define the random variable $\tilde{X}_i \equiv X_{\alpha(i)}$, where
$$
\alpha(i)\equiv \sum_{j=1}^{M(t)} \alpha_j(i) 1_{\{X_{[1]} = X_j \} }.
$$
The distribution of $\bar{X}_d(\sigma_2)$ is the joint distribution of $(\tilde{X}_1, \ldots, \tilde{X}_{M(\sigma_2)-1})$, which characterizes the size distribution of the clones in the distant field at time $\sigma_2$. We obtain the following result.
\begin{theorem}\label{coro1}
The  size-distribution of the distant field clones at time $\sigma_2$ of the first successful type-2 mutation, conditioned on $\{\sigma_2=t\}$, is given by
\begin{align*}
\mathcal L( \bar{X}_d|\in dt) =_d \hat{P}(\tilde{X}_1 &\in dx_1, \ldots, \tilde{X}_{M(t)-1} \in dx_{M(t)-1}) \\
&=\frac{1}{1-e^{-\lambda t\phi(t)}} \sum_{m=1}^\infty \frac{(\lambda \phi(t)t)^{m} e^{- \lambda \phi(t) t}}{m!} \prod_{i=1}^{m-1} g_t(x_i),
\end{align*}
where $g_t(x)$ is defined in \eqref{cond_Xi}.
\end{theorem}
\begin{proof}
See section \ref{pr5}.
\end{proof}
\begin{remark} From  Theorem \ref{coro1} and Corollary \ref{lemma2} below, we see that
\begin{align*}
\mathcal L( \bar{X}_d|\sigma_2=t, M(t)=m) =_d \hat{P}(\tilde{X}_1 &\in dx_1, \ldots, \tilde{X}_{m-1} \in dx_{m-1})=\prod_{i=1}^{m-1} g_t(x_i).
\end{align*}
\end{remark}

 \begin{figure}[htb!]
   \centering
   \includegraphics[width=0.6\textwidth]{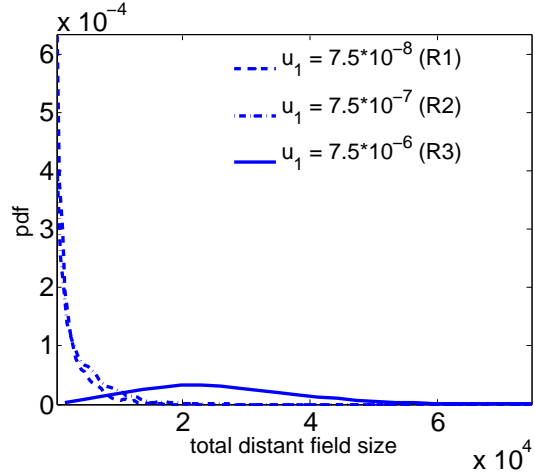} 
   \caption{The distribution of the total size of the distant field is shown for different scenarios, corresponding to the three regimes R1, R2 and R3 illustrated in Figure \ref{fig:regimes} for varying type-1 mutation rates $u_1$.  The non-varying parameters are held constant at  $d=2$, $N=2\cdot 10^5$, $u_2=2\cdot 10^{-5}$, $s_1=s_2=0.1$ and $c_2(s_1)=0.16$.} 
   \label{fig:distant_pdf}
\end{figure}

Figure \ref{fig:distant_pdf} shows how the probability density function of the total distant field size (i.e.\ the sum of all distant field patches) changes with increasing mutation rate $u_1$.   For a comparison to the local field size distribution at the same parameter values, we refer to Figure \ref{fig:sbp_compilation}B.  We note that in regimes 1 and 2 the total distant field size is on the same order of magnitude as the local field size, but in regime three the distant field size is significantly larger than the size of the local field.  As will be investigated in more detail below, this suggests that secondary tumor recurrences for cancer types in regime 3 are much more likely to stem from the distant field, and thus are more likely to be clonally unrelated to the primary tumor.

\subsection{Number of field patches: evolution until initiation}
We next analyze the total number of premalignant lesions over time until tumor initiation.  In particular, the following result holds:
\begin{proposition} Conditioned on $\{\sigma_2=t\}$, we have that for all $\zeta\leq t$, the number of field patches is distributed as a mixture of a Poisson and a shifted Poisson random variable.  In particular,
\begin{align*}
\begin{split}
 P(M(\zeta)=m| \sigma_2 =t)&= p_1(t,\zeta) \, \frac{\lambda^{m} \cb{t\phi(t) - (t-\zeta)\phi(t-\zeta)}^{m}}{(m)!}e^{-\lambda \cb{t\phi(t) - (t-\zeta)\phi(t-\zeta)}}\\
& \qquad +p_2(t,\zeta)\, \frac{\lambda^{m-1} \cb{t\phi(t) - (t-\zeta)\phi(t-\zeta)}^{m-1}}{(m-1)!}e^{-\lambda \cb{t\phi(t) - (t-\zeta)\phi(t-\zeta)}},
\end{split}
\end{align*}
where $p_1(t,\zeta)+p_2(t,\zeta)=1$ and $p_1(t,\zeta) = (1-e^{-\theta (t-\zeta)^{d+1}})/(1-e^{-\theta t^{d+1}})$. In particular,
\begin{align*}
 E(M(\zeta)|\sigma_2=t)=\lambda \cb{t\phi(t) - (t-\zeta)\phi(t-\zeta)} + p_2(t,\zeta).
\end{align*}
\label{clone_dist}
\end{proposition}
\begin{proof}
To prove this result we calculate $P(\sigma_2>t|M(\zeta)=m)$ and then differentiate with respect to $t$. Details of the proof are given in section \ref{pr6}.
\end{proof} It is interesting to observe that as $\zeta\to t$ we see that $p_1(t,\zeta)\to 0$, therefore as $\zeta$ gets closer to time $t$ the process looks more like a shifted Poisson. This is stated in the corollary below.

\begin{corollary}\label{lemma2}
\begin{align}\label{res1}
\hat P(M(t)=m)=  \frac{\ob{\lambda \, t\, \phi(t)}^{m-1}}{(m-1)!}e^{-t \lambda \phi(t)}, \quad m\geq 1, 
\end{align}
and $\hat P(M(t)=m)=0$. In particular,
 \begin{align}\label{hier}
 \hat E (M(t))= 1+  E(M(t)| \sigma_2>t) =1+ \lambda t \phi(t) ,
 \end{align}
 where $E(M(t)|\sigma_2>t)$ is discussed in Lemma \ref{lemma1}.
 
\end{corollary}

Using Proposition \ref{clone_dist}, we can study the expected number of field patches of a certain size over time. Figure \ref{fig:temp_clones} shows the temporal dynamics of clone-size distribution in each regime.  In regime 1 the expected number of small clones peaks and then declines as larger clones begin to dominate (consistent with the notion that a single premalignant clone exists prior to initiation), whereas in regimes 2 and 3 we see longer coexistence of large and small clones over time.  
\begin{figure}[htb!]
   \centering
   \includegraphics[width=1\textwidth]{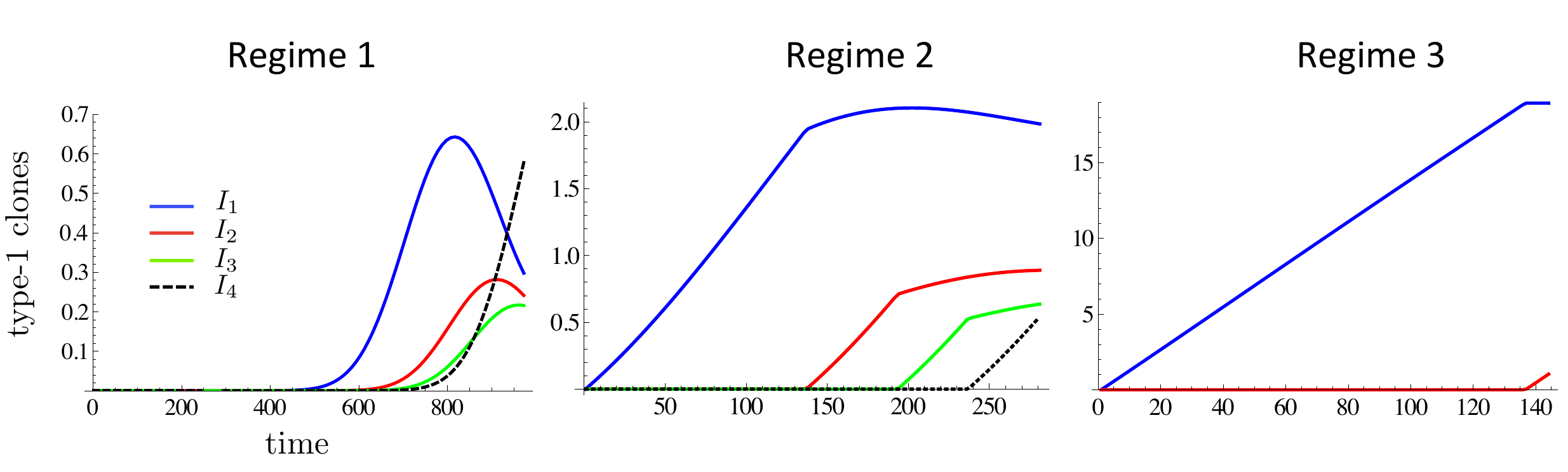} 
   \caption{ {\bf Dynamic clone-size distribution.} For each of the three regimes in Figure \ref{fig:wait_time}, the expected number  of type-1 clones of sizes comprised in the corresponding intervals $I_j$ are shown as  functions of time up to $E(\sigma_2)$ (expectations are conditioned on $\{t=E(\sigma_2)$\}. The intervals are defined as $I_1=[0,1500)$, $I_2=[1500,3000)$, $I_3=[3000,4500)$ and $I_4=[4500,+\infty)$. Parameter values as in Figure \ref{fig:wait_time}.
} 
   \label{fig:temp_clones}
\end{figure}

\begin{remark}\label{remark:conditional_process}
The result in Proposition \ref{clone_dist} can be extended to a result about the entire process $\{M(r):0\leq r\leq t\}$ conditioned on $\sigma_2=t$. We present here the joint distribution of the process at multiple time points, since the proof is similar to Proposition \ref{clone_dist} we do not include it.
For $0\leq r\leq r^\prime\leq t$ define 
$$
\hat\phi(t;r,r^\prime)=\int_r^{r^\prime}e^{-\theta(t-y)^{d+1}}dy.
$$
Then for any positive integer $\ell$, sequence of time points $0<r_1\leq \ldots\leq r_{\ell}<t$ and non-negative integers $k_1\leq k_2\leq \ldots\leq k_{\ell}$ we have that 
\begin{align*}
&\hat{P}\left(M(r_1)=k_1,\ldots,M(r_{\ell})=k_{\ell}\right)\\
&=
\left(\sum_{i=1}^{\ell}\frac{k_i-k_{i-1}}{\hat{\phi}(t;r_{i-1},r_i)}p_i+\lambda p_{\ell+1}\right)\frac{1}{\lambda}\prod_{j=1}^{\ell}\frac{\left(\lambda\hat\phi(t;r_{j-1},r_j)\right)^{k_j-k_{j-1}}}{(k_j-k_{j-1})!}e^{-\lambda\hat\phi(t;r_{j-1},r_j)},
\end{align*}
where for $1\leq i\leq \ell+1$,
$$
p_i=\frac{e^{-\theta(t-r_i)^{d+1}}-e^{-\theta(t-r_{i-1})^{d+1}}}{1-e^{-\theta t^{d+1}}},
$$
$r_0=0$, $k_0=0$, and $r_{\ell+1}=t$. Note that for each $i$, $0<p_i<1$ and $\sum_{i=1}^{\ell+1}p_i=1$, i.e.\ the $p_i$'s form a probability vector.

The above joint distribution is rather difficult to parse, so we describe how one would generate samples of the increments of the process. For $1\leq i\leq \ell$, set $X_i=M(r_i)-M(r_{i-1})$, then we can generate the values of the vector $X_1,\ldots, X_{\ell}$ under the measure $\hat{P}$ as follows.  For each $1\leq i\leq\ell$ sample $X_i$ according to a Poisson distribution with mean $\lambda\hat\phi(t;r_{i-1},r_i)$. Choose an integer $I$ according to the probability vector $(p_1,\ldots,p_{\ell+1})$, if $I=i<\ell+1$ replace $X_i$ with $X_i+1$.

Note that in contrast to the setting of a Poisson process the random variables $X_1,\ldots, X_{\ell}$ are not independent under $\hat{P}$.
\end{remark}

\section{Recurrence predictions}\label{recur_pred}

Tumor recurrence due to field cancerization poses a substantial clinical problem in many epithelial cancers \cite{chai2009field}. We next aim to use the results of the previous section to develop a methodology for assessing the risk of tumor recurrence (as well as the likely type of tumor recurrence) after surgical removal of the primary tumor. 
\subsection{Local vs. distant field recurrence?}\label{locfieldrec}
As discussed above, a recurring tumor can either arise in the same premalignant field (a second field tumor), or it can arise in a clonally unrelated field (second primary tumor). In this section we  characterize the recurrence time distribution for each of these secondary tumor types, and study how the relative likelihood of local vs. distant recurrence depends upon parameters of the tissue and cancer type.
 
To this end, we first study the recurrence time distribution for second field tumors, which arise from the local premalignant field.  Denote the second field recurrence time by $T_R^f$, measured in time units $\tau$ starting from $\tau=0$ at time $\sigma_2$. The time is reset at the tumor initiation time $\sigma_2$, rather than the tumor resection time $\sigma_2 +T_D$, to accommodate the possibility that a recurrence occurs {\it prior} to detection of the primary tumor.  Thus if recurrence occurs at some time $\tau < T_D$, then a secondary tumor already exists at the time of diagnosis of the primary tumor (but may be too small to be detectable).  We assume that the primary tumor node is completely resected once it becomes detectable at time $T_D$, leaving the surrounding field intact (i.e.\ there are no excision margins).


At time $\sigma_2$ a successful type-2 cell arises from a premalignant clone of radius $R_l(\sigma_2)$, whose distribution is characterized in Corollary \ref{Cor_r}.  If $R_l(\sigma_2) = r$, the incidence rate of successful type 2 mutations within this field is given by
\begin{align}\label{this}
\eta(r,\tau)\equiv u_2 \bar s_2 \gamma_d \cb{\ob{r + c_d(s_1) \tau}^d - \bar{c}_d^d(s_2)\ob{\tau\wedge T_D}^d},
\end{align}
where $\bar{c}_d(s_2)$ is the rate of expansion of the malignant cells into the type-1 field.  A new variable $\bar{c}_d(s_2)$  is introduced to account for the fact that the malignancy may grow both upwards into the epithelial layer and sideways into the field, depending on the type of tissue involved (see also Figure \ref{fig:epi_cartoon}B). 
\begin{corollary}\label{cori}
The probability of a second field tumor having formed before time $\tau$ (measured from $\sigma_2$), conditioned on $\{\sigma_2 =t\}$, is given by 
$$
\hat P(T_R^f<\tau) =1- \frac{\gamma_d  u_2 \bar{s}_2 }{c_d(s_1)(1-e^{-\theta t^{d+1}})}  \int_0^{c_d(s_1) t} r^d\exp \left[ -\frac{u_s\bar{s}_2\gamma_d}{c_d(s_1)(d+1)} r^{d+1} -\int_0^\tau \eta(r,s)ds\right] dr.
$$ 
 In particular, $\hat P(T_R^f<T_D)$ is the probability that smaller, possibly undetectable second field tumors exist at the time of diagnosis.
\end{corollary}
\begin{proof}
See section \ref{pr9}.
 \end{proof}
In Figure \ref{fig:local_rec}A  the cumulative distribution function of $T_R^f$ as calculated in Corollary \ref{cori} is shown, for varying values of type-2 mutation rates $u_2$. As one might expect, higher mutation rates yield a decreased time to recurrence (the curves shift to the left for increasing $u_2$). However, considering that the size of the premalignant field at initiation of the primary tumor is inversely proportional to the mutation rate $u_2$, see Figure \ref{fig:local_rec}B, the decrease in time to recurrence is {\it a priori} not obvious: a bigger precancer field increases the chance of fast recurrence. This example illustrates how a quantitative model enables us to assess the relative importance of competing aspects of the system  - in this case, the impact of larger premalignant field versus higher mutation rates on recurrence likelihood.
 \begin{figure}[htb!]
   \centering
  \includegraphics[width=1\textwidth]{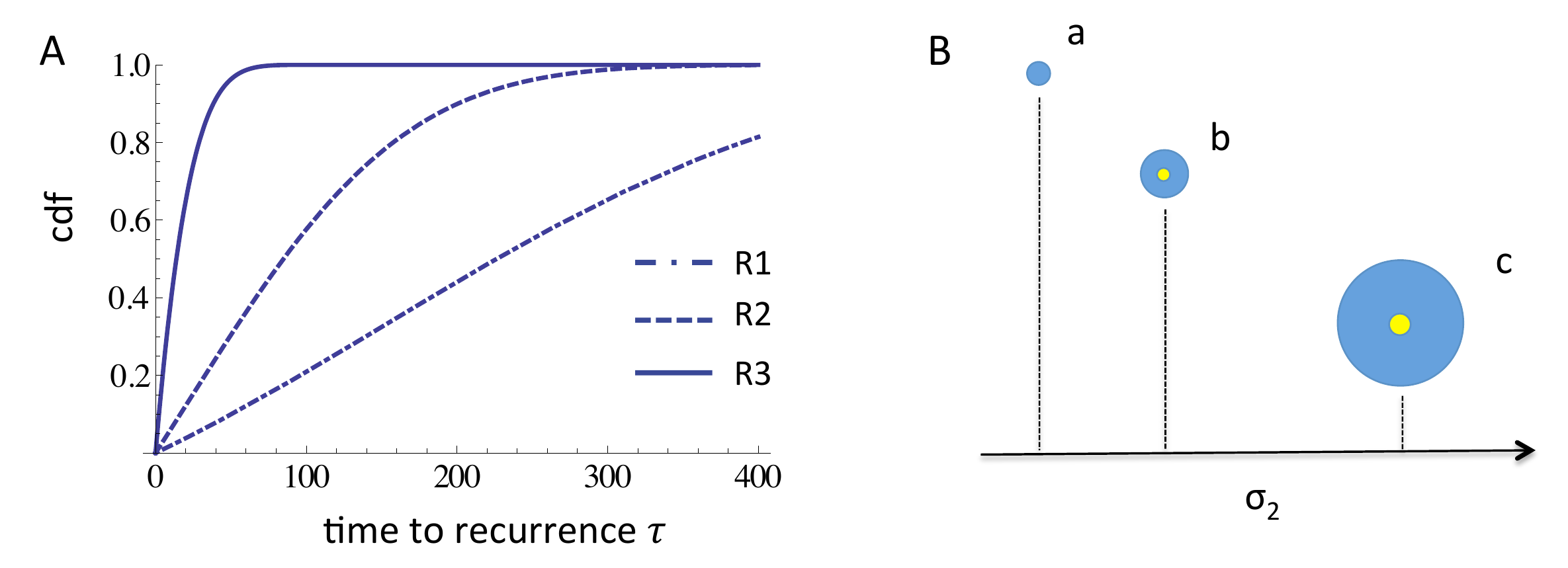} 
   \caption{ {\bf Time to local recurrence.} {\bf A} The cumulative distribution function of the time to recurrence of a second field tumor is shown for three different scenarios, corresponding to $u_2 = 2\cdot 10^{-3}$ (Regime 1), $u_2 = 2\cdot 10^{-5}$ (Regime 2) and $u_2 = 2\cdot 10^{-3}$ (Regime 2/3), respectively.  The remaining parameters are  $d=2$, $N=2\cdot 10^5$, $u_1=7.5\cdot 10^{-7}$,  $s_1=s_2=0.1$,  $t=E(\sigma_2)$. {\bf B} Schematic of the relative initiation times of the primary tumor (yellow) and sizes of the local fields (blue), for the three scenarios in panel A. The numerical values for expected initiation time and local field size are: {\bf (a)} $\E(\sigma_2)=123$, $\hat E(R_l)=8$; {\bf (b)} $\E(\sigma_2)=281$, $\hat E(R_l)=31$; {\bf (c)} $\E(\sigma_2)=474$, $\hat E(R_l)=55$. } 
   \label{fig:local_rec}
\end{figure} 



\vspace{0.1in}
If the recurrence does not take place in the local field giving rise to the first successful type-2 clone, then it either arises from one of the type-1 clones already present at time of initiation (i.e.\ the distant field), or it arises in a type-1 clone formed after initiation. In the latter case, the waiting time is again distributed as $\sigma_2$, and hence we focus here on the  distribution of the waiting time $T_R^p$, defined as the time from $\sigma_2$ until a second primary tumor arises from the distant field already existing at $\sigma_2$. We have the following result.
\begin{corollary}\label{dist_recur}
The probability that the distant field at the time of initiation gives rise to a second primary tumor by time $\tau$ (measured from $\sigma_2$), conditioned on $\{\sigma_2 =t\}$, is given by 
\begin{align*}
 P(T_R^p>\tau |  \sigma_2=t) &   =  \exp\cb{-\lambda t \phi(t) \ob{1-d \gamma_d \Phi(\tau,t)}}
\end{align*}
where
\begin{align*} \Phi(\tau,t)=\int_0^\infty \exp\ob{-\int_0^\tau \eta \ob{r,s} ds} r^{d-1}  g_t(r_i ^d  \gamma_d) dr,
\end{align*}
and $g_t$ is defined in \eqref{cond_Xi}.
\end{corollary}

\begin{proof}
See section \ref{pr10}.
\end{proof}

Thanks to the results in this section, it is now possible to evaluate the probability of local versus distant tumor recurrences in each parameter regime.  Corollary \ref{cori} explicitly provides the probability density function  $\hat{P}(T_R^f \in d\tau)$, which is the probability that a second field tumor arises at time $\tau$ from the same field that gave rise to the primary tumor.  To obtain the corresponding probability density function for recurrence as a second primary tumor, we have to consider recurrences due to distant field lesions that have arisen before and after $\sigma_2$. While Corollary \ref{dist_recur} characterizes the recurrence risk due to distant lesions already present at initiation, the time to a successful second primary tumor from a distant field not yet present at initiation is distributed as $\sigma_2$, see Theorem \ref{pdfsig2}.  Therefore, the distribution of interest is that of $\tilde T_R^p=\min\{T_R^p,\sigma_2\}$, which is the time of the first distant recurrence event. 

In Figure \ref{fig:local_vs_distant} we study how the comparison between the probability density functions of $T_R^f$ (second field tumor, local) and $\tilde T_R^p$ (second primary tumor, distant) varies in regimes 1, 2 and 3. The likelihood of local vs. distant recurrences depends strongly upon both the timing and parameter regime of the system  In regime 1, local recurrence is significantly more likely overall, but at late times the probability of distant recurrences is slightly higher than for local recurrences.  In contrast, in regimes 2 and 3 the overall probability of local and distant recurrences are comparable.  However, in regime 2, at early times distant field recurrences are more likely, whereas the opposite is true at later times. The same observation, but even more pronounced, holds in regime 3.  These results show that fundamental tissue parameters (e.g.\ renewal rate, size, mutation rate) could be used to provide insights into the timing and clonal relatedness of tumor recurrences, and suggest a new strategy for prognosis prediction and risk stratification.


  \begin{figure}[htb!]
   \centering
  \includegraphics[width=.7\textwidth]{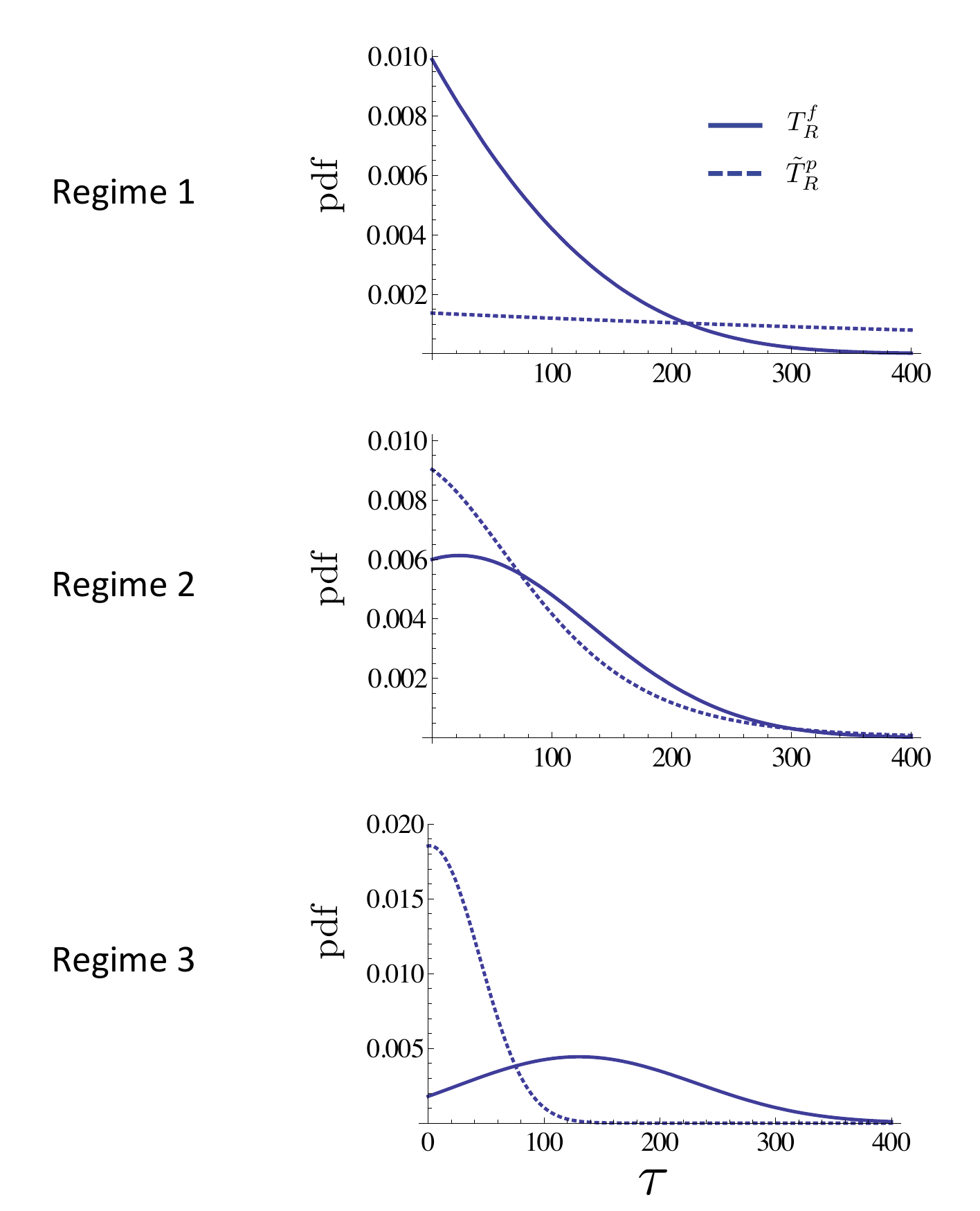} 
   \caption{ {\bf Local vs. distant recurrence.} {\bf A} For each of the three regimes in Figure \ref{fig:wait_time}, we show: the distribution of time to local recurrence $\hat P(T_R^f\in d\tau)$, and the distribution of time to distant recurrence $\hat P(\tilde T_R^p\in d \tau)$. The  distribution of $T_R^f$ is given in Corollary \ref{cori} and we set $\tilde T_R^p=\min\{T_R^p,\sigma_2\}$  to account both for contributions from type-1 clones already existing at $\sigma_2$ as well as contributions from type-1 clones born after $\sigma_2$ (for which time to recurrence is distributed as $\sigma_2$). Expected times to recurrence: $\hat E(T_R^f)=81$ and $\hat E(\tilde T_R^p)=733$ (Regime 1);   $\hat E(T_R^f)=98$ and $\hat E(\tilde T_R^p)=86$ (Regime 2);  $\hat E(T_R^f)=149$ and $\hat E(\tilde T_R^p)=34$ (Regime 3). The parameter values are as in Figure \ref{fig:wait_time}. } 
   \label{fig:local_vs_distant}
\end{figure}

\section{Conclusions and outlook}

In this study we performed a quantitative analysis of the cancer field effect by means of a spatial stochastic model of cancer initiation, which had previously been introduced in \cite{DuFoLe12}. 
Based on this model, we investigated various characteristics of premalignant fields at the time of tumor initiation. In particular, we derived the size-distributions of the local field (the premalignant lesion that gives rise to the tumor) and the distant field (the premalignant lesions that are unrelated to the primary tumor). We calculated the  dynamic clone size distribution at times leading up to initiation, and derived the probability density functions of local and distant recurrence times. Finally, we compared the relative likelihood of second field versus second primary tumors, and 
 demonstrated how the clonal relatedness between primary and recurrent tumors depends explicitly upon tissue and cancer type parameters.

Using an example set of parameters in two space dimensions (which is appropriate for describing the cancer initiation process in the basal layer of a stratified epithelium), we found that lower mutation rates (such as in regime 1)  were associated with larger local field sizes, whereas higher mutation rates (regimes 2 and 3) led to smaller local fields. We also found that higher mutation rates resulted in larger distant fields, while more aggressive cancers (high selective advantage) led to larger local fields at diagnosis. Finally, we investigated the risk of recurrence after surgical resection of the malignant portion, and found that for low mutation rates (regime 1), local recurrence is much more likely, whereas for larger mutation rates (regimes 2 and 3), the overall probability of local and distant recurrences are comparable.  However, in regimes 2 and 3, early recurrences are more likely to be a second primary tumor, whereas the late recurrences are more likely to be second field tumors.


One important limitation of our approach is that the model captures a specific sequence of genetic alterations with specified $u_i$ and $s_i$, and does currently not allow for permutations of genetic events and divergent pathways. Nevertheless, our model may provide a useful framework for comparing different biological hypotheses and disentangling divergent genetic pathways among cancer subtypes. In particular, it enables us to predict differences in observable dynamics such as initiation times and prognoses between different molecular models. Such an approach could help elucidating the sequence of genetic events during carcinogenesis, and will be the subject of future work.  Another limitation of our framework is that we have assumed a static, uniform microenvironment within the tissue. The local microenvironment is in reality determined by a variety of time- and space-dependent factors such as glucose, oxygen, growth factors, drugs and cytokine concentrations. In addition to impacting the growth and mutation rates of cells within the tissue, the local microenvironment is increasingly being recognized as playing an important role in carcinogenesis through stromal signaling.  Lastly, in order to apply this framework to specific cancer types (as will be done in a follow-up work), appropriate parameter values for the model need to be determined.  Literature estimates can usually be obtained for tissue parameters such as compartment size and tissue regeneration rates.  In addition, baseline estimates for point mutation rates are available, usually down to an uncertainty of 1-2 orders of magnitude.  Estimates of the relative selective advantage conferred by each mutation are more difficult to ascertain; however some estimates can be obtained from proliferation marker staining and {\it in vitro} studies.

In summary, the analyses performed in this work contribute towards a quantitative understanding of how organ-specific physiological parameters (e.g.\ number of proliferative cells, tissue renewal rates), as well as pathway-specific parameters (e.g.\ cellular mutation rate, selective advantages conferred by each oncogenic mutation) influence the process of field cancerization and the associated risks of recurrence. We demonstrated that tumor recurrence dynamics and premalignant field characteristics are strongly dependent upon these parameters, which vary across  different tissue and cancer types.  Once properly calibrated for a specific tissue and cancer type, the proposed methodology can potentially be used to provide insights into key prognostic factors such as risk of multifocal lesions and tumor recurrence,  surveillance guidelines, and treatment design. For example, we are able to assess the likelihood and timing of local versus distant recurrences after surgical resection.  Since this distinction provides information on the level of clonal relatedness between primary and recurrent tumors, the model predictions may provide insights into whether treatment strategies effective for primary tumors will be useful for recurrent tumors in particular cancer types.  In addition, our methodology can be utilized to assess the relative benefits of surgical excision margins (removal of apparently healthy tissue surrounding tumors), and to help determine the minimal margins necessary to prevent recurrence in each tissue type. We consider the work presented here as a {\it proof of concept}: a mathematical framework which, once properly calibrated, refined and validated, may provide a useful tool in molecular epidemiology (e.g.\ mechanistic modeling of carcinogen exposure), as well as the development of probabilistic models for personalized treatment approaches. 


\section{Acknowledgements}
We thank Rick Durrett for insightful discussions on this project as well as his useful  
suggestions on the manuscript.
\section{Appendix: Proofs} \label{s_pfs}

\subsection{Proof of Theorem \ref{sbp_thm}}\label{app_loc_field}
To prove Theorem \ref{sbp_thm}, we first need a few new definitions and preliminary results.  Define $V(t)$ to be the random total space-time volume covered by successful type-1 families until time $t$,
\begin{align}
V(t) = \sum_{i=1}^{M(t)} \gamma_d c_d^d(s_1) \frac{(t-T_i)^{d+1}}{d+1},
\label{eq:Vrand}
\end{align}
where  $T_i$ represents the arrival time of the $i$-th family, and $M(t)$ is the total number of successful arrivals by time $t$, which is a Poisson process with rate $\lambda$.  Let $\vet$ represent the space-time volume conditioned on the event 
$$\mathcal{E}_{t}(t_1, \ldots, t_m) \equiv \{ M(t)=m, T_1 \in dt_1, \ldots T_m \in dt_m\},$$ where $0 < t_1 < \cdots < t_m < t$. In other words, 
\begin{align}
V_{\mathcal{E}_t}\equiv\frac{\gamma_dc_d^d(s_1)}{d+1}\sum_{i=1}^m(t-t_i)^{d+1}.
\label{eq:Vdet}
\end{align}
For ease of notation we replace $V_{\mathcal{E}_{t}(t_1, \ldots, t_n)}$ with the more compact version $V_{\mathcal{E}_t}$. 
Since $E[V(t)] = E[E[V(t) | M(t)]]$ and the conditioned process is a compound Poisson process, we obtain that
$$
E[V(t)] = \sum_{m=0}^\infty P(M(t)=m) \frac{m \gamma_d c_d^d(s_1)}{d+1} E[(t-T_i)^{d+1}] = \lambda \gamma_d c_d^d(s_1) \frac{t^{d+2}}{(d+2)(d+1)}.
$$
Similarly, we define $A(t)$ to be the total area of clones covered by successful type-1 families at time $t$,
\begin{align} A(t) \equiv \sum_{i=1}^{M(t)} \gamma_d c_d^d(s_1)(t-T_i)^d,
\label{eq:Arand}
\end{align} and we define $\aet$ to be this quantity conditioned on $\Ec$,
\begin{align} \aet \equiv \sum_{i=1}^{m} \gamma_d c_d^d(s_1)(t-t_i)^d.
\label{eq:Adet}
\end{align}
Note that
\begin{align}\label{EA}
E[A(t)] = \sum_{m=0}^\infty P(M(t)=m) m \gamma_d c_d^d(s_1) E[(t-T_i)^d] = \lambda \gamma_d c_d^d(s_1) \frac{t^{d+1}}{d+1}.
\end{align}

By considering the space-time volume of type-1 clones we can calculate $P(\sigma_2>t|\mathcal{E}_t(t_1,\ldots,t_m)$ and $P(\sigma_2>t|M(t)=m)$. Combining these two formulas and using Bayes rule we get the following result for the joint distribution of the arrival times of successful type-1 mutations, conditioned on the total number of mutations by time $t$.  
\begin{lemma}\label{lemlem}
\label{iid_ti}
Conditioned on $\set{\sigma_2>t}$ and $\set{M(t)=m}$, the arrival times of successful type-1 clones $(T_1, \ldots, T_m)$  are distributed as order statistics of iid random variables as follows:
\begin{align*}
P(T_1\in dt_1,\ldots, T_m \in dt_m |\sigma_2>t, M(t)=m) = \frac{m!}{t^m \phi(t)^m}\prod_{i=1}^m e^{-\theta(t-t_i)^{d+1}}
\end{align*}
where $0 < t_1 < \cdots < t_m < t$.
\end{lemma}

\begin{proof} The arrival process of successful type-1 mutations is represented by $M(\cdot)$, which is a Poisson process with rate $\lambda = Nu_1 s_1/(1+s_1)$ and arrival times $T_1,T_2,\ldots$.   Then for any $t>0$ and sequence $0<t_1<\cdots<t_m<t$ we have that
\begin{align}\label{hier}
P(\Ec)=\lambda^me^{-\lambda t}.
\end{align}
Since
\begin{align}
P(\sigma_2>t | \Ec)=\exp(-u_2\bar{s}_2\vet),
\label{eq:sigma_2_cond}
\end{align}
we find using Bayes' rule
$$
P(\sigma_2>t,\Ec)=\lambda^me^{-\lambda t}\exp(-u_2\bar{s}_2 \vet).
$$
It follows then that
\begin{align*}
&P(T_1\in dt_1,\ldots, T_m \in dt_m |\sigma_2>t, M(t)=m)=\frac{P(\sigma_2>t,\Ec)}{P(\sigma_2>t|M(t)=m)P(M(t)=m)}\\
&=
\frac{\lambda^m e^{-\lambda t}\exp\left(-u_2\bar{s}_2\vet\right)}{P(\sigma_2>t|M(t)=m)e^{-\lambda t}(\lambda t)^m/m!}\\
&=
\frac{m!}{t^m}\frac{\exp\left(-u_2\bar{s}_2\vet\right)}{\left( E\exp\left(-u_2\bar{s}_2\gamma_dc_d^d(s_1)(t-T)^{d+1}/(d+1)\right)\right)^m}\\
&=
m!\prod_{i=1}^m\left(\frac{1}{t}\right)\frac{\exp\left(-u_2\bar{s}_2\gamma_dc_d^d(s_1)(t-t_i)^{d+1}/(d+1)\right)}{E\exp\left(-u_2\bar{s}_2\gamma_dc_d^d(s_1)(t-T)^{d+1}/(d+1)\right)},
\end{align*}
where $T$ is a uniform random variable on $[0,t]$. .
\end{proof}
 
The distribution in Lemma \ref{lemlem} is an exponential twist of the uniform  distribution. Note that if the conditioning was placed on the set $\set{\sigma_2=t}$ instead of $\{\sigma_2>t\}$, then the conditional distribution would no longer have product form because of the term
$
\frac{d}{dt}\vet,
$
 and the arrival times would not be the order statistics from an iid collection of random variables.

Next, we show that the random variable $M(t)$ is Poisson if conditioned on $\set{\sigma_2>t}$.
\begin{lemma}\label{lemma1}
Conditioned on $\{\sigma_2>t\},$ $ M(t) =_d \text{Pois} \ob{\lambda t \phi(t)}.$
\end{lemma}
\begin{proof}  First we note that 
\begin{equation}\label{hh1}
\begin{split}
P(\sigma_2>t)= &\sum_{m=0}^\infty \frac{1}{m!} \int_{[0,t]^m} P ( \sigma_2>t |\mathcal E_t (t_1,\ldots, t_m)) P(\mathcal E_t (t_1,\ldots, t_m)) dt_1\ldots dt_m  \\
=& \sum_{m=0}^\infty \frac{1}{m!} \int_{[0,t]^m} \exp\ob{-u_s \bar s_2 V_{\mathcal E_t}} \lambda^m e^{-\lambda t}dt_1\ldots dt_m \\
=&\sum_{m=0}^\infty \frac{1}{m!} t^m \lambda^m \ e^{-\lambda t}\ob{\frac{1}{t}\int_0^t \exp\ob{-\frac{u_2 \bar s_2  \gamma_d c_d^d(s_1) (t-r)^{d+1}}{d+1} }dr}^m  \\
=& \sum_{m=0}^\infty \frac{\ob{t \lambda \phi(t)}^m}{m!} e^{-\lambda t}= e^{t\lambda(\phi(t)-1)}.
\end{split}
\end{equation}
From this,  we find using Bayes' rule 
\begin{equation}\label{eq22}
\begin{split}
P(\mathcal E_t (t_1,\ldots, t_m) | \sigma_2 >t ) =& \frac{P(\sigma_2 >t|\mathcal E_t (t_1,\ldots, t_m) \, P(\mathcal E_t (t_1,\ldots, t_m))}{P (\sigma_2>t)}\\
=& \frac{\lambda ^m e^{-\lambda t} \exp(u_s \bar{s}_2 V_{\mathcal E_t)}}{e^{t \lambda( \phi(t)-1)}},
\end{split}
\end{equation}
and hence
\begin{align*}
P(M(t)= m| \sigma_2>t )=& \frac{1}{m!}\int_{[0,t]^m} P(\mathcal E_t(t_1,\ldots,t_m)| \sigma_2>t ) dt_1\ldots dt_m\\=& e^{-\lambda t \phi(t))} \, \frac{\ob{t\lambda \phi(t)}^m}{m!}.
\end{align*}.
\end{proof}

For subsequent considerations, it will be useful to define the two conditional probability measures $\hat{P}(\cdot)=P(\cdot|\sigma_2 =t)$ and $\tilde{P}(\cdot)=P(\cdot|\sigma_2>t)$, and their corresponding expected values,  $\hat{E}(\cdot)=E(\cdot|\sigma_2=t)$ and $\tilde{E}(\cdot)=E(\cdot|\sigma_2>t)$, respectively. In particular, we can compute the Radon-Nikodym derivative between these two measures. 
 \begin{lemma}\label{lem2} The Radon-Nikodym derivative of $\hat{P}$ with respect to $\tilde P$ is given by
 \begin{align}\label{RND}
\frac{d\hat{P}}{d\tilde{P}} = \frac{\aet u_2 \bar{s}_2}{\lambda (1 - e^{-\theta t^{d+1}})}.
\end{align}
 \end{lemma}
 \begin{proof} First, note that
 \begin{align}\label{phat}
 P(\Ec|\sigma_2 =t) = \frac{P(\Ec) P(\sigma_2 =t | \Ec)}{P(\sigma_2 =t)}.
 \end{align}
By differentiating  \eqref{eq:sigma_2_cond} and \eqref{hh1} we obtain
$$
P(\sigma_2 =t |  \Ec) = u_2\bar{s}_2 \aet\exp(-u_2\bar{s}_2\vet)
$$
and
\begin{align}
\label{dens_sig2}
P(\sigma_2 \in  dt) &= -\frac{d}{dt}e^{t\lambda(\phi(t)-1)} = \lambda \ob{1- e^{-\theta t^{d+1}}}e^{t\lambda(\phi(t)-1)}.
\end{align} 
Hence (\ref{phat}) becomes
$$
 P(\Ec|\sigma_2 =t)  = \lambda^m e^{-\lambda t} \frac{ u_2 \bar{s}_2 \aet \exp(-u_2\bar{s}_2\vet )}{\lambda e^{t\lambda(\phi(t)-1)} (1-e^{-\theta t^{d+1}})},
$$
and comparing this to (\ref{eq22}) yields the desired result.
\end{proof}
 

Recall now that $M(t)$ is the number of successful type-1 mutations that have arrived by time $t$, and we denote their arrival times by $T_1,\ldots, T_{M(t)}$. At time $t$, the area of a clone created at  time $r<t$  is $\gamma_dc_d^d(s_1)(t-r)^d$, and hence the area of the $i$-th clone at time $t$ is given by the random variable
$$
X_i(t) \equiv \gamma_dc_d^d(s_1)(t-T_i)^d. 
$$
Using the above results together with definition \ref{sbp} of a size-biased pick we can now prove Theorem \ref{sbp_thm}.

\begin{proof}[Proof of Theorem \ref{sbp_thm}]
Using basic properties of conditional expectations and Definition \ref{sbp} we find
\begin{align*}
\hat{P}(X_{[1]}\in dx)&=\hat{E}\left[\hat P(X_{[1]}\in dx|X_1,\ldots,X_{M(t)},M(t))\right]\\
&=
\hat{E}\left[\sum_{i=1}^{M(t)}\frac{X_i1_{\{X_i\in dx\}}}{S_{M(t)}}\right]=\sum_{m=1}^{\infty}\hat{E}\left[\sum_{i=1}^{m}\frac{X_i1_{\{X_i\in dx\}}}{S_{m}}1_{\{M(t)=m\}}\right],
\end{align*}
where $S_m=X_1+\ldots+X_m$.
Using the Radon-Nikodym derivative (\ref{RND}) we can rewrite this as
\begin{equation}\label{fmla_sbdens}
\begin{split}
&=\sum_{m=1}^\infty \tilde{E} \left[  \frac{1_{\{M(t)=m\}}u_2 \bar{s}_2}{\lambda (1 - e^{-\theta t^{d+1}})} \left( \sum_{i=1}^m \frac{X_i  1_{\{ X_i \in dx \}}}{S_m}\right)   \sum_{j=1}^m X_j  \right] \\ 
&= \frac{u_2 \bar{s}_2}{\lambda (1 - e^{-\theta t^{d+1}})} \sum_{m=1}^\infty \tilde{E} \left[ 1_{\{ M(t)=m \}}\sum_{i=1}^m x 1_{\{ X_i \in dx \}}\right] \\ 
&= \frac{xu_2 \bar{s}_2}{\lambda (1 - e^{-\theta t^{d+1}})} \sum_{m=1}^\infty E\left[\sum_{i=1}^m 1_{\{X_i \in dx\}} | M(t) = m, \sigma_2 > t  \right] P(M(t)=m | \sigma_2 > t)   \\
&= \frac{xu_2 \bar{s}_2}{\lambda (1 - e^{-\theta t^{d+1}})} P(X_1 (t)\in dx | M(t) = m, \sigma_2 > t) E\left[ M(t) | \sigma_2 > t\right],
\end{split}
\end{equation}
where we have used the fact that $P(X_1(t) < x | M(t) = m, \sigma_2 > t)$ is independent of $m$, which we will show below. 
Using Lemma \ref{iid_ti} and differentiating the cumulative distribution function
\begin{align*}
P(X_1(t) < x | M(t) = m, \sigma_2 > t) &= P\left(T_1 > t - \left( \frac{x}{\gamma_d c_d^d(s_1)} \right)^{1/d} \Big| M(t) =m, \sigma_2 > t\right),
\end{align*}
we determine that 
\begin{align}
\label{cond_Xi}
P(X_1 (t)\in dx | M(t) = m, \sigma_2 > t) =\frac{x^{1/d - 1}}{d \gamma_d^{1/d} c_d(s_1)t  \phi(t)} \exp\left[\frac{-u_2\bar{s}_2 x^{\frac{d+1}{d}}}{(d+1)\gamma_d^{1/d} c_d(s_1) } \right] \equiv g_t(x)
\end{align}
for $x \in [0, \gamma_d c_d^d(s_1) t^d]$.  Note that (\ref{cond_Xi}) is indeed independent of $m$.  From Lemma \ref{lemma1} it follows that 
$$
E\cb{M(t)|\sigma_2>t} = \lambda t\phi(t),
$$
and combined with \eqref{fmla_sbdens} and  \eqref{cond_Xi} this yields the desired result.
\end{proof}

\subsection{Proof of Theorem \ref{coro1}}\label{pr5}
Using Definition \ref{sbp} of a size-biased pick we find
\begin{align*}
\hat{P}&(\tilde{X}_1 \in dx_1, \ldots, \tilde{X}_{M(t)-1} \in dx_{M(t)-1})\\
&=\hat{E}[\hat P(\tilde{X}_1 \in dx_1, \ldots, \tilde{X}_{M(t)-1} \in dx_{M(t)-1} | X_1, \ldots, X_{M(t)}, M(t))]\\
&= \hat{E} \left[ \sum_{j=1}^{M(t)} \frac{X_j}{S_{M(t)}} \prod_{i=1}^{M(t)-1}1_{ \{ X_{\alpha_j(i)} \in dx_i\} } \right]\\
&= \frac{u_2 \bar{s}_2}{\lambda (1 - e^{-\theta t^{d+1}})} \sum_{m=1}^\infty P(M(t)=m | \sigma_2 > t) E\left[\sum_{j=1}^m X_j  \prod_{i=1}^{m-1}1_{ \{ X_{\alpha_j(i)} \in dx_i\} } \Big| \sigma_2 > t, M(t)=m \right],
\end{align*}
where the final equality follows from the same sequence of arguments as used in the proof of  Theorem \ref{sbp_thm}.  Next, we note that 
\begin{align*}
E[X_j(t)| \sigma_2 > t, M(t)=m ] = & \int_0^\infty x P(X_j (t)\in dx | M(t) = m, \sigma_2 > t)= \int_0^\infty x \, g_t(x) dx \\ 
= & \int_0^{\gamma_d c_d^d(s_1) t^d} \frac{x^{1/d }}{d \gamma_d^{1/d} c_d(s_1) \phi(t) t} \exp\left[\frac{-u_2\bar{s}_2 x^{\frac{d+1}{d}}}{(d+1)\gamma_d^{1/d} c_d(s_1) } \right] dx\\
=& \frac{1}{\phi(t)t u_2 \bar s_2}\cb{1-\exp\ob{-\frac{u_2 \bar s_2 \gamma_d c_d^d(s_1) t^{d+1}}{d+1}}}, 
\end{align*}
and 
$$
 \sum_{j=1}^mE\left[ X_j  \prod_{i=1}^{m-1}1_{ \{ X_{\alpha_j(i)} \in dx_i\} } \Big| \sigma_2 > t, M(t)=m \right] = \sum_{j=1}^m E[X_j| \sigma_2 > t, M(t)=m ] \prod_{i=1}^{m-1} g_t(x_i).
$$
Together with Lemma \ref{lemma1} the result follows.

\subsection{Proof of Proposition \ref{clone_dist}}\label{pr6}
First, we use Bayes' rule to find
\begin{align}
P(\Ecs| \sigma_2=t)= \frac{P(\sigma_2\in dt| \Ecs) \, P(\Ecs)}{P(\sigma_2\in dt)}.
\end{align}
Since $P(\sigma_2\in dt)$ is given in (\ref{dens_sig2}) and $P(\Ecs)=\lambda^me^{-\lambda \zeta}$, it remains to calculate $P(\sigma_2\in dt| \Ecs)$. It is easy to see that 
\begin{align}\label{ja}
P(\sigma_2>t|\Ecs)= \exp\ob{-u_2 \bar s_2 \vet} q(\zeta,t),
\end{align}
where $q(\zeta,t)$ is the probability that a type-2 mutation arises in a clone that is born in the interval $(\zeta,t)$. We find  
\begin{align*}
 q(\zeta,t)= &E\cb{e^{- \theta \sum_{i=1}^{M(t-\zeta)} (t-T_i)^{d+1}}}\\
 =& E\cb{E \cb{\left. e^{- \theta \sum_{i=1}^{M(t-\zeta)}(t-T_i)^{d+1}}\right | M(t-\zeta)}}\\
 = &E \cb{\phi(t-\zeta)^{M(t-\zeta)}}=e^{\lambda(t-\zeta) (\phi(t-\zeta)-1)},
\end{align*}
where the last expression is the generating function for the Poisson process. Together with (\ref{ja}) this yields now
\begin{align*}
P(\sigma_2\in dt | \Ecs)=& - \frac{d}{dt}P(\sigma_2>t|\Ecs)\\
=& e^{\lambda(t-\zeta) (\phi(t-\zeta)-1)}e^{-u_2 \bar s_2 \vet } \, \cb{u_s \bar s_2 \aet+ \lambda\ob{1- e^{-\theta (t-\zeta)^{d+1}}}}
\end{align*}
Together with (\ref{dens_sig2}) and (\ref{hier}), we find now
\begin{align*}
P(\Ecs| \sigma_2=t)= \lambda^{m-1}  \frac{e^{-\lambda \cb{t\phi(t))-(t-\zeta) \phi(t-\zeta)}}}{ \ob{1-e^{-\theta t^{d+1}}}} e^{-u_2 \bar s_2 \vet } \cb{u_s \bar s_2 \aet+ \lambda\ob{1- e^{-\theta (t-\zeta)^{d+1}}}},
\end{align*}
and hence performing the integration in
\begin{align*}
\hat  P (M(\zeta)=m)= \int_{[0,\zeta]^m} \frac{1}{m!} P(\Ecs | \sigma_2=t) dt_1\ldots dt_m
\end{align*}
yields the desired result.

\subsection{Proof of Corollary \ref{cori}}\label{pr9}
 \begin{align*}
 \hat P(T_R^f>\tau)=& P\ob{T_R^f>\tau| \sigma_2 =t}\\
 =& \int_0^{c_d(s_1) t} P(T_R^f>\tau, R_{l}(\sigma_2) \in  dr| \sigma_2=t)dr\\
 =&\int_0^{c_d(s_1) t} P(T_R^f>\tau | R_{l}(\sigma_2) \in  dr, \sigma_2=t) P(R_l (\sigma_2) \in dr |\sigma_2 =t) dr,
 \end{align*}
 where $R_l(t)$ is the radius of the local field surrounding the tumor at time $t$. The result follows from 
 \begin{align}
 P(T_R^f>\tau | R_{l}(\sigma_2) \in dr, \sigma_2 =t) = \exp\ob{-\int_0^\tau \eta(r,s)ds}
 \end{align}
and the  conditional density of $R_l(\sigma_2)$ in Corollary \ref{Cor_r}.

\subsection{Proof of Corollary \ref{dist_recur}}\label{pr10}
First, we note that
\begin{align}\label{nomau}
\begin{split}
&P(T_R^p>\tau | M(t)= m, \sigma_2=t) \\
 & \qquad =\int_{\mathbb{R}_{+}^{m-1}} P\ob{T_R^p>\tau | \tilde R_1\in dr_1, \ldots, \tilde R_{m-1} \in dr_{m-1}, M(t)=m, \sigma_2=t}\cdots \\
 & \qquad \qquad \cdots  P\ob{ \tilde R_1\in dr_1, \ldots, \tilde R_{m-1} \in dr_{m-1}| M(t)=m, \sigma_2=t} ,
 \end{split}
\end{align}
where $\tilde R_i$ are the radii of the distant field clones, corresponding to their respective areas $\tilde X_i$ defined in Section \ref{distfieldinit}. Recalling the definition of $\eta$ in (\ref{this}), we find
\begin{align}\label{momou}
P\ob{T_R^p>\tau | \tilde R_1\in dr_1, \ldots, \tilde R_{m-1} \in dr_{m-1}, M(t)=m, \sigma_2 \in dt} = \exp\ob{-\sum_{i=1}^{m-1} \int_0^\tau \eta \ob{r_i,s} ds}.
\end{align}
Recalling the Radon-Nikodym derivative $d\hat P/d\tilde P$ from Lemma \ref{lem2}, it is straight-forward to verify that 
\begin{align*}
\frac{dP\ob{t_1, \ldots t_m | M(t)=m, \sigma_2=t}}{dP\ob{t_1, \ldots t_m | M(t)=m, \sigma_2>t} } =&\frac{d\hat P}{d\tilde P}\,\frac{P(M(t)=m | \sigma_2>t) }{P(M(t)=m | \sigma_2 =t )} = \frac{\aet u_2 \bar{s}_2 t \phi(t)}{ m(1 - e^{-\theta t^{d+1}})},
\end{align*}
which allows us to derive the following expression (proceeding as in the proof of Corollary \ref{coro1}),
\begin{align*}
& P\ob{ \tilde X_1\in dx_1, \ldots, \tilde X_{m-1}\in dx_{m-1} | M(t)=m, \sigma_2=t}= \prod_{i=1}^{m-1} g(x_i)dx_i.
\end{align*}
Switching from the clone-areas $\tilde X_i$ back to the corresponding radii $\tilde R_i$, we find 
\begin{align*}
 P&\ob{ \tilde R_1\in dr_1, \ldots, \tilde R_{m-1}\in dr_{m-1} | M(t)=m, \sigma_2 =t} = 
 (d\,\gamma_d)^{m-1}   \prod_{i=1}^{m-1} r_i^{d-1}g(r_i ^d \gamma_d) dr_i
\end{align*}
From this, (\ref{momou}) and (\ref{nomau}) we find
\begin{align}
P\ob{T_R^p > \tau | M(t)=m, \sigma_2=t}=\ob{d\, \gamma_d\, \Phi(\tau,t)}^{m-1},
\end{align}  
Finally, using Lemma \ref{lemma1}, 
\begin{align*}
\hat P(T_R^p >\tau)& = \sum_{m=1}^\infty P\ob{T_R^p > \tau | M(t)=m, \sigma_2=t} \hat P(M(t)=m)\\
&= \exp\ob{-\lambda t \phi(t) \ob{1-d \gamma_d \Phi(\tau,t)}}
\end{align*}

\end{document}